\documentclass[journal]{IEEEtran}
\usepackage{lipsum}
\usepackage{color}
\usepackage{cite}
\ifCLASSINFOpdf
\usepackage[pdftex]{graphicx}
\usepackage{epstopdf}
\usepackage{amsthm}
\usepackage{amsmath,amssymb,amsfonts}
\ifCLASSOPTIONcompsoc
\usepackage[caption=false,font=normalsize,labelfont=sf,textfont=sf]{subfig}
\else
\usepackage[caption=false,font=footnotesize]{subfig}
\fi
\usepackage{stfloats}
\newtheorem{remark}{Remark}
\newtheorem{theorem}{Theorem}
\newtheorem{assumption}{Assumption}

\newtheorem{problem}{Problem}
\begin{document}
\title{Recovery of Power Flow to Critical Infrastructures using Mode-dependent Droop-based Inverters}
\author{Soham~Chakraborty,~\IEEEmembership{Student Member,~IEEE,}
        Sourav~Patel,~\IEEEmembership{Member,~IEEE,}
        and~Murti V. Salapaka,~\IEEEmembership{Fellow,~IEEE}
\thanks{S. Chakraborty, S. Patel, M. V. Salapaka are with the Department of Electrical and Computer Engineering, University of Minnesota, Minneapolis, 55455 MN, USA~ e-mail: chakr138@umn.edu, patel292@umn.edu, murtis@umn.edu}
\thanks{The authors acknowledge Advanced Research Projects Agency-Energy (ARPA-E) for supporting this research through the project titled ``Rapidly Viable Sustained Grid'' via grant no. DE-AR0001016.}
}
\maketitle
\begin{abstract}
Recovery of power flow to critical infrastructures, after grid failure, is a crucial need arising in scenarios that are increasingly becoming more frequent. This article proposes a power transition and recovery strategy by proposing a mode-dependent droop control-based inverters. The control strategy of inverters achieves the following objectives 1) regulate the output active and reactive power by the droop-based inverters to a desired value while operating in on-grid mode 2) seamless transition and recovery of power flow injections into the critical loads in the network by inverters operating in off-grid mode after the main grid fails; 3) require minimal information of grid/network status and conditions for the mode transition of droop control. A framework for assessing the stability of the system and to guide the choice of parameters for controllers is developed using control-oriented modeling. A comprehensive controller hardware-in-the-loop-based real-time simulation study on a test-system based on the realistic electrical network of M-Health Fairview, University of Minnesota Medical Center, corroborates the efficacy of the proposed controller strategy.
\end{abstract}
\section{Introduction}\label{intro}
\IEEEPARstart{C}{ritical} infrastructure (CI) describes essential assets of society that includes medical centers and hospitals, security service centers, food production and distribution centers, communication infrastructures. Disruption of power to CIs often result in a debilitating impact on physical and economic security, public health and safety. Rapid and seamless recovery of power flow, possibly after a power blackout caused by weather/climate disasters, to restore CIs online is a crucial need arising in scenarios that are increasingly becoming more frequent \cite{NOAA}. Accreditation standards such as IEEE $602$ require CIs to have emergency power supply system (EPSS) in order to form a local microgrid network with local generation sets and automatic transfer switches (ATSs), in case of sudden power blackouts of main grid supply \cite{ieee602}. Depending on the level of criticality and urgency of the electrical loads, the EPSS is required to form the local microgrid within permissible time in order to restore the operation of CIs. Various types of EPSS include, Type-$\mathrm{U}$ that designates uninterruptible EPSS and Type-$10$ where the EPSS is allowed $10$~s for recovery \cite{nfpa110}). Gas/diesel generator sets are traditional choices for most of the CIs as local energy units for EPSS due to their sustained and robust power supply capability. However, the long startup time from standby mode makes the task of a seamless and rapid power restoration of CIs difficult to achieve with the aforementioned sources \cite{ieee446}. 
\par Battery storage unit interfaced with power electronic inverters, usually termed as uninterrupted power supply (UPS), is an alternate solution that enhances the ease in operation and reduces the response time of EPSS for CIs. Accreditation standards such as NFPA $111$ recommend stored-energy EPSS (SEPSS) that employ batteries/fuel-cells/ultra-capacitors as main energy harvesting units along with voltage source inverter (VSI) topology to assist in restoration of power to CIs in case of grid failure \cite{nfpa111}. In order to achieve rapid and seamless recovery of power flow for Type-$\mathrm{O}$/Type-$\mathrm{U}$ SEPSS that demand no electrical interruption, VSIs are required to be always synchronized and connected with the power system network of CI irrespective of the availability of main grid, unlike the conventional plug-and-play strategies \cite{csCIu}. However, functionality of remaining synchronized and active with the
power system network adds challenges in operation of VSIs. Here, SEPSS while ensuring that the VSIs are connected to the system, need to guarantee that the VSI do not alter the normal operation of the CI and  should provide no power when the grid is available. Thus in on-grid mode, all of the power to the CI is to be supplied only by the distribution grid, with the VSIs remaining on standby to enable a seamless transition to an off-grid mode in case of distribution grid interruption. In case of a failure of the distribution grid, SEPSS is required to ensure that the CI can function while maintaining a stable voltage and frequency and meeting the power demand by the locally stored energy units via VSIs in off-grid mode. Here, while on-grid it is crucial for SEPSS to maintain sufficient reserves of energy in battery storage units for emergency off-grid operation. 
\par Droop controller-based autonomous and communication-less approach for paralleling multiple battery-fed VSIs in off-grid mode is an effective decentralized strategy \cite{microgrid1}. Reference \cite{chandorkar} introduced the concept of \textit{conventional droop controller} for multiple-UPS-fed off-grid and dominantly inductive microgrid network by emulating the behavior of a synchronous generator in the classical power system. For other network conditions that arise in power distribution systems, several modifications on droop control are proposed \cite{viability,gurrero1,gurrero2,gurrero3,gurrero4,gurrero5,gurrero6,gurrero7}. Some representative references, that emphasize improved power sharing capabilities, are \cite{ghosh,zhong1,zhong2,lampke,bullo,misc1,misc2}. However, all these works are primarily restricted to multi-VSI-based microgrids operating in an off-grid mode. During the on-grid mode, the voltage and frequency of the network are governed by the stiff grid. Here, unlike the off-grid operation, the output active and reactive power of the droop-controlled VSIs will be heavily influenced by the stiff distribution grid where there is a need for the active and reactive power references to be adjusted dynamically \cite{microgrid1,davodi}. Reference \cite{droopgrid1} proposes an adaptive droop control for VSIs suitable for both on- and off-grid mode of operation of microgrid. However, knowledge of magnitude, type of grid impedance and coupling impedances of VSIs (by impedance estimation techniques) are prerequisite for this control which may not be practical in distribution systems where the parameters keep changing. Master-slave-based architecture in multi-VSI systems (electrically closest VSI to grid as master and rests as slaves) is proposed in \cite{droopgrid2,droopgrid3} both for on- and off-grid mode. However, a coordinated architecture suffers from the loss of autonomy and independent nature of operation of multi-VSIs. Reference \cite{droopgrid5} proposes a modified version of droop control law for VSIs in order to achieve operation in on-grid mode; here, the prime focus is to inherit the advantages of the droop controller to limit the inverter current under both normal and faulty grid conditions.
\begin{figure}[t]
	\centering
    \includegraphics[scale=0.26,trim={0cm 0cm 0cm 0cm},clip]{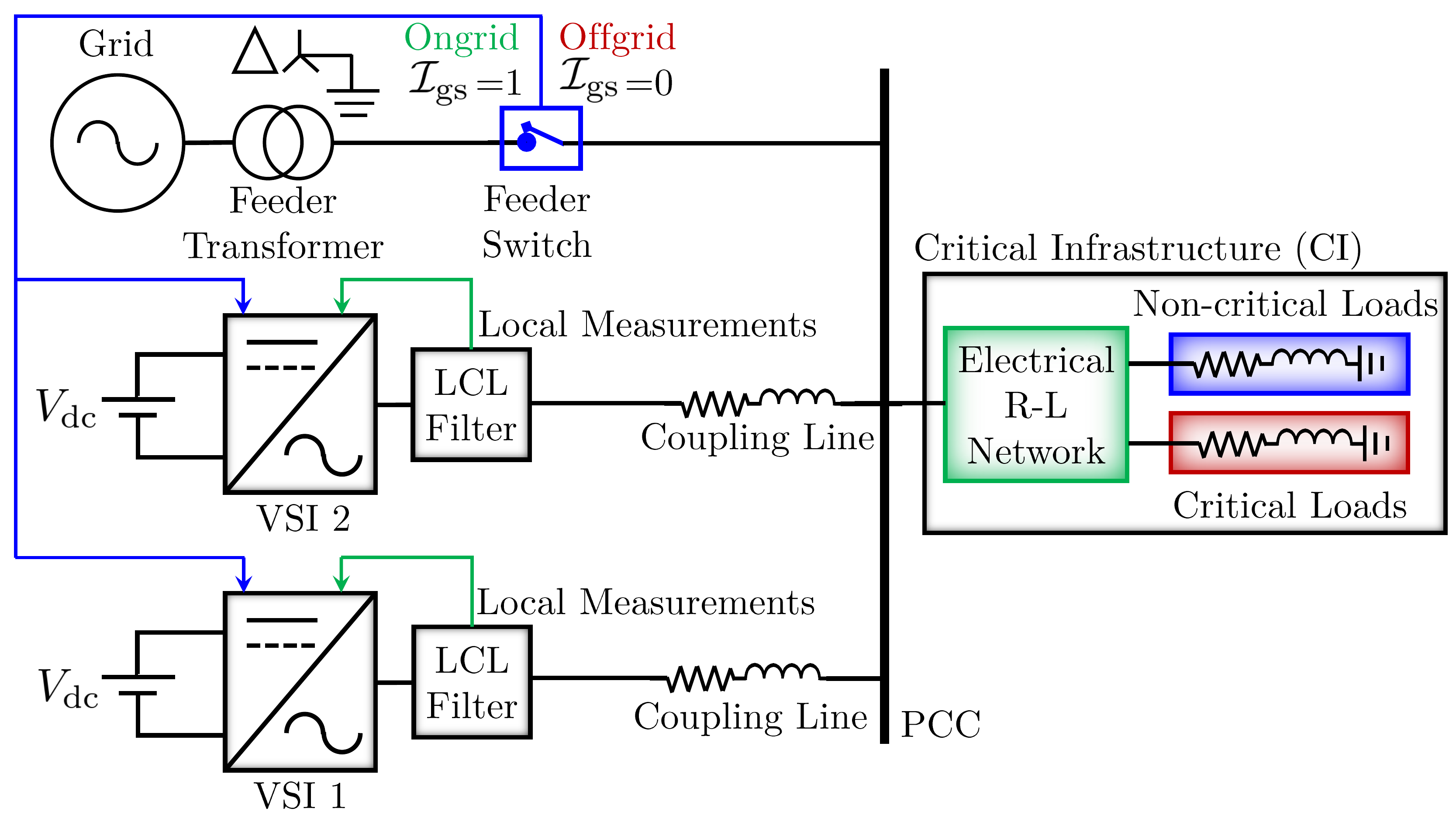}%
	\caption{A critical infrastructure with electrical R-L network and critical/non-critical loads supplied by grid, VSI-$1$ and VSI-$2$ at PCC.}
	\label{fig:system1}
\end{figure}
\par In order to mitigate the above issues, a mechanism for multi-VSI-based systems is proposed that enables both on- and off-grid operation while providing seamless transitions between on-grid and off-grid operations with good regulation of output power. This article develops a novel mode-dependent droop control framework for VSIs with the following features; 1) regulation over supplied output active and reactive power of the VSIs to the desired value (zero reference values in this application) while operating in on-grid mode; 2) a fast response time of recovery of generating sources of SEPSS to the critical loads in the CI by VSIs operating in off-grid mode once the main grid fails; 3) requirement of minimal additional information of grid/network status and conditions for the mode transition of droop control. While the objectives $1$ and $2$ are achieved by the proposed droop control for both on-grid and off-grid modes, objective $3$ is met by requiring only the information (as a single bit $0$ or $1$) on grid availability. Grid availability can be detected employing standard island detection techniques such as the remote island detection techniques. Supervisory remote island detection, as proposed in \cite{island}, is adopted in this work for its fast and accurate performance. With this proposed mode-dependent droop control for SEPSS it is shown, both analytically and experimentally, that VSIs will remain synchronized to the grid and can be regulated to supply no active and reactive power to the network while operating in on-grid mode while the entire load of CI is supplied by the utility grid. Whereas, during off-grid mode the VSIs share the required critical load demand among themselves while exhibiting a seamless transition from on-grid mode after grid failure to act as primary source of generations for the CI microgrid. Moreover, a systematic modeling of the entire system is carried out in order to assess the stability of the system and to guide the selection of the parameters needed for the proposed mode-dependent droop controllers. In order to evaluate the efficacy the proposed approach, controller hardware-in-the-loop-based real-time simulation studies are conducted on a test-system developed based on the realistic electrical network of M-Health Fairview, University of Minnesota Medical Center. Results validate the viability and the performance of the proposed design.
\par This article is organized as follows. In Section \ref{motivation}, the motivation and the description of the system under study are presented. In Section \ref{control}, the proposed control architecture is described for individual VSI. In Section \ref{model}, individual components of the test system under study are modeled which are used for analysis. Section \ref{analysis} is focused on the analysis of system performance with the proposed controller where the objective of the system stability guides the choice of control parameters. Section \ref{result} shows the experimental setup and corresponding results. Finally, Section \ref{conclusion} concludes the article.
\section{Description of System and Motivation}\label{motivation}
The system, considered in this work, comprises a CI as shown in Fig.~\ref{fig:system1}. The considered CI has two classification of loads as defined in NFPA $99$ \cite{nfpa99} namely, i) critical loads: essential load services and facility loads that need continuous uninterrupted power supplied irrespective of the grid availability and ii) non-critical loads: utility services and facilities where temporary shut-down of services is allowed in case of grid failures. In the off-grid mode, the CI unit operates supported by two VSIs as shown in Fig. \ref{fig:system1}. At the emergency bus, which is the point of common coupling (PCC), multiple sources are connected in parallel and supply the load demand of the CI. The grid, interfaced with feeder transformer and feeder switch, is the main power supply in on-grid mode of the CI feeding both critical and non-critical loads. VSI-$1$ and VSI-$2$ are always synchronized and connected to the PCC by coupling lines irrespective of grid availability. Both VSIs have output voltage and current measurements available, while operating with proposed mode-dependent droop controller as described in Section~\ref{control}. In case of a grid failure, SEPSS needs to ensure that supply to critical load demand is met by VSI-$1$ and VSI-$2$ which are the primary sources of CI in off-grid mode. Depending on the status of the feeder switch (on/off), SEPSS will ensure to switch the modes of proposed droop-controller of the VSIs by means of the transmitted status signal referred to as $\mathcal{I}_\mathrm{gs}$ in the rest of the article ($\mathcal{I}_\mathrm{gs}=1$ in on-grid mode and $\mathcal{I}_\mathrm{gs}=0$ in off-grid mode). Design objectives of the SEPSS are as follows:
\begin{itemize}
    \item Provide rapid and seamless power recovery where VSIs are required to be always connected to the CI network at PCC while SEPSS ensures not to de-energize these energy sources,
    \item During normal scenarios where the distribution grid is available, CI is required to be supported only by the grid and both critical and non-critical load demands needs to be met,
    \item In case of a grid failure, SEPSS is required to ensure that the CI must meet critical load demands by the local energy resources in off-grid mode. 
\end{itemize}
In the following section, the control architecture and proposed mode-dependent droop control for individual VSI is presented.
\section{Design of Single Droop-controlled Inverter}\label{control}
\begin{figure}[t]
	\centering
    \includegraphics[scale=0.24,trim={0cm 0cm 0cm 0cm},clip]{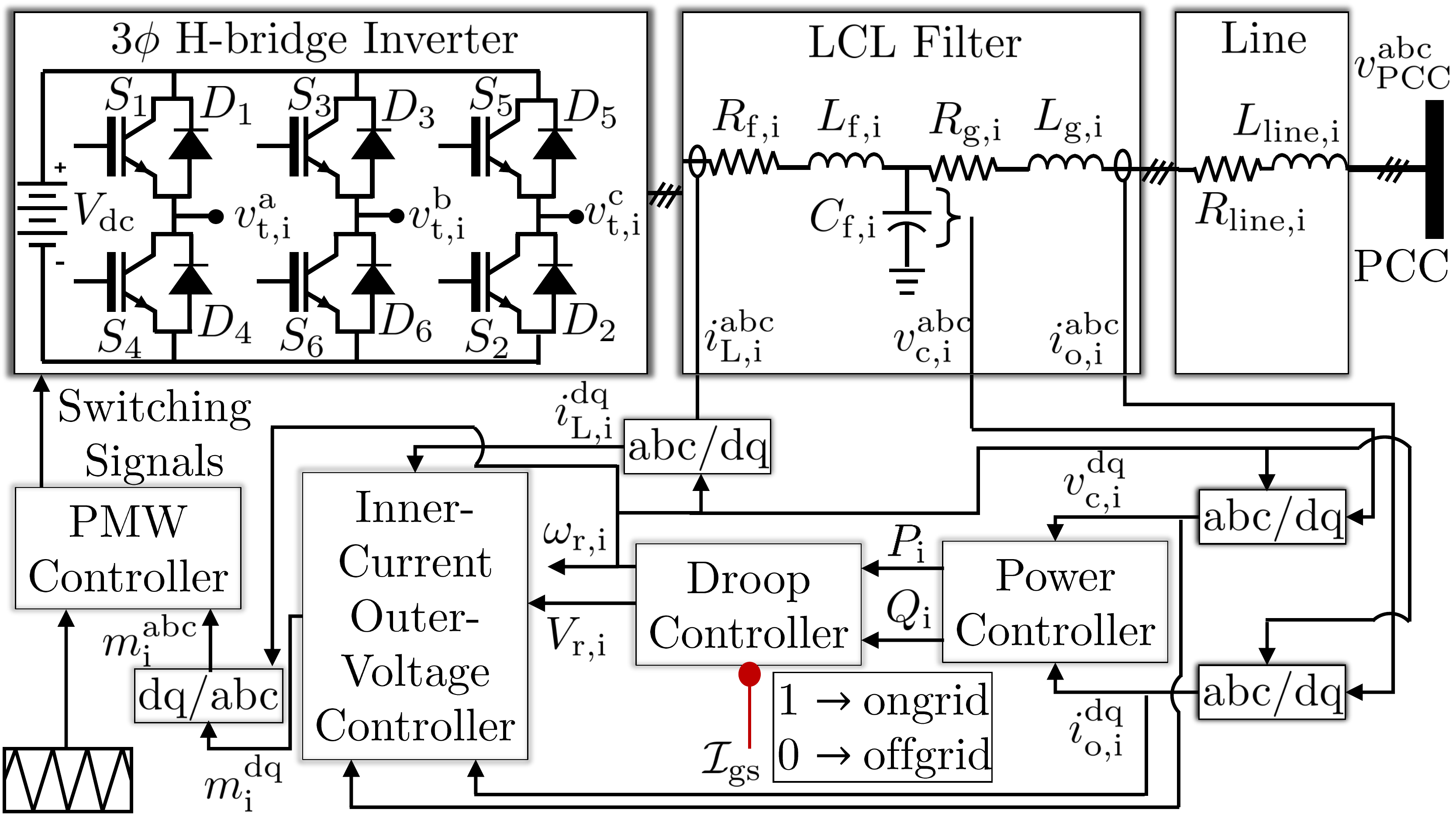}%
	\caption{The power and control block diagram of $\mathrm{i}^{th}$ VSIs connected to PCC.}
	\label{fig:circuit}
\end{figure}
The power circuit of $\mathrm{i}^{th}$, $3$-$\phi$ H-bridge VSIs consist of six switches distributed among three legs as shown in Fig.~\ref{fig:circuit}. It is connected to the network at PCC with voltage, $v^\mathrm{abc}_\mathrm{PCC}$, via an $\mathrm{LCL}$ filter ($L_\mathrm{f,i}$, $C_\mathrm{f,i}$, $L_\mathrm{g,i}$ and associated equivalent series resistances (ESRs) $R_\mathrm{f,i}$ and $R_\mathrm{g,i}$ of inductors) and a coupling line (with line parameters $L_\mathrm{line,i}$, $R_\mathrm{line,i}$) as shown in Fig.~\ref{fig:circuit}. A $\mathrm{dq}$-frame multi-loop structured controller is adopted that generates modulated voltage vector signal, $m^\mathrm{abc}_\mathrm{i}$, to pulse-width modulation (PWM) controller to produce switching signals for the power switches resulting in terminal voltages, $v^\mathrm{a}_\mathrm{t,i}$, $v^\mathrm{b}_\mathrm{t,i}$ and $v^\mathrm{c}_\mathrm{t,i}$. The control loop of Fig.~\ref{fig:circuit} is described briefly below.
\subsubsection{Power Controller}\label{power}
The $\mathrm{d}\mathrm{q}$-axis (w.r.t. $\mathrm{i}^{th}$ VSI reference frame) output voltage ($v^\mathrm{dq}_\mathrm{o,i}$) and current ($i^\mathrm{dq}_\mathrm{o,i}$) measurements are used to determine the instantaneous active power ($p_\mathrm{i}$) and reactive power ($q_\mathrm{i}$) supplied by the inverter\footnote{Notation: $x^\mathrm{abc}$ is defined as $\begin{bmatrix}x^\mathrm{a}& x^\mathrm{b}&x^\mathrm{c}\end{bmatrix}^\top$ and $x^\mathrm{dq}$ is defined as $\begin{bmatrix}x^\mathrm{d}& x^\mathrm{q}\end{bmatrix}^\top$ where (.)$^\top$ denotes transposition.}. $p_\mathrm{i}$ and $q_\mathrm{i}$ are passed through low-pass filters with the time constant, $\tau_\mathrm{S,i} \in \mathbb{R}_{>0}$, to obtain the filtered average output active and reactive power according to
\begin{align}\label{powercal}
    P_\mathrm{i} =  \dfrac{1}{\tau_\mathrm{S,i}s+1}p_\mathrm{i},~Q_\mathrm{i} =  \dfrac{1}{\tau_\mathrm{S,i}s+1}q_\mathrm{i}.
\end{align}
where $p_\mathrm{i}:=\dfrac{3}{2}[v^\mathrm{d}_\mathrm{o,i}i^\mathrm{d}_\mathrm{o,i} + v^\mathrm{q}_\mathrm{o,i}i^\mathrm{q}_\mathrm{o,i}]$ and $q_\mathrm{i}:=\dfrac{3}{2}[v^\mathrm{q}_\mathrm{o,i}i^\mathrm{d}_\mathrm{o,i}-v^\mathrm{d}_\mathrm{o,i}i^\mathrm{q}_\mathrm{o,i}]$.\vspace{1mm}
\begin{figure}[t]
	\centering
    \includegraphics[scale=0.24,trim={0cm 0cm 11cm 7cm},clip]{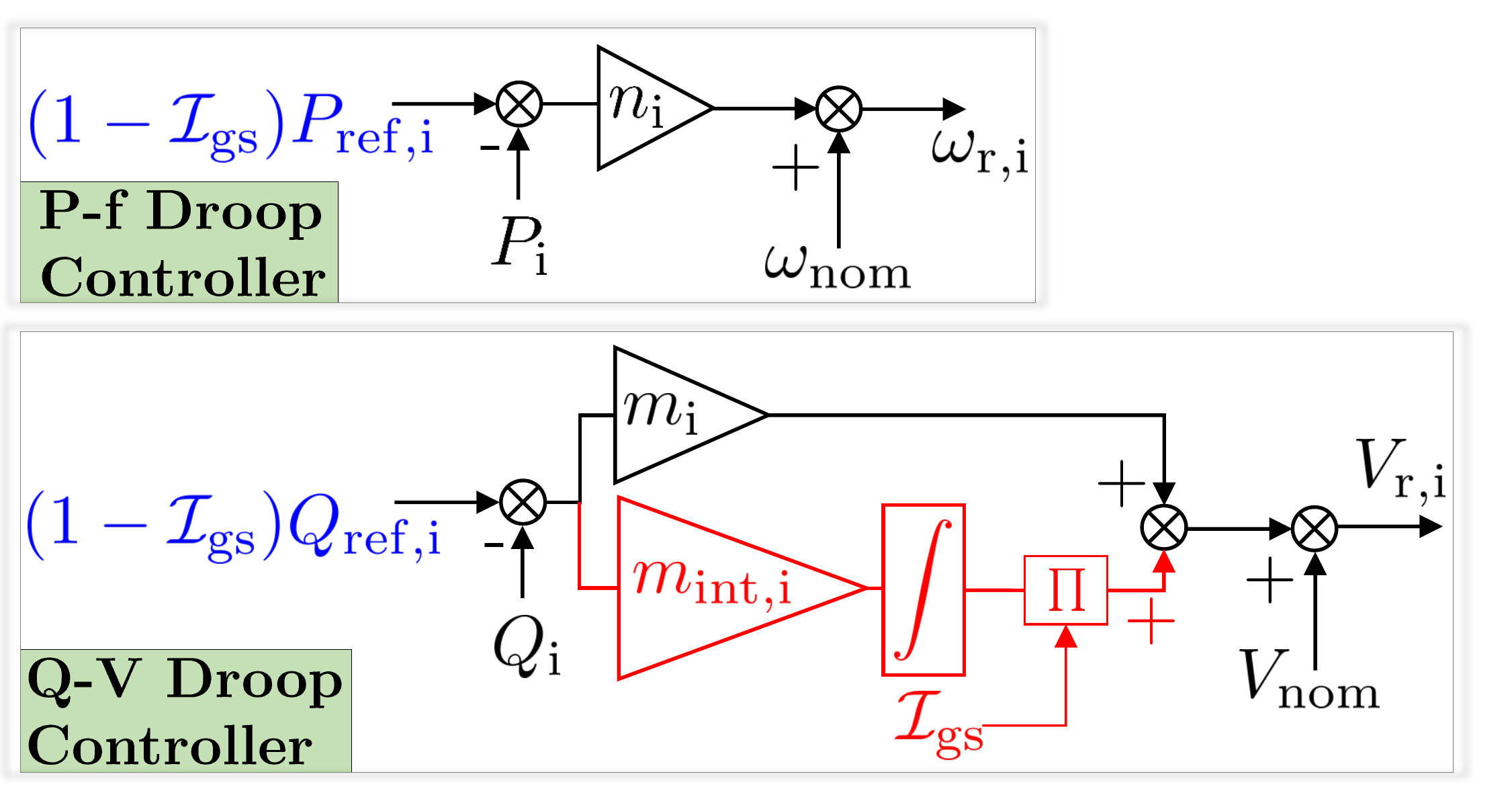}%
	\caption{The proposed mode-dependent droop controller for $\mathrm{i}^{th}$ VSI}
	\label{fig:droop}
\end{figure}
\subsubsection{Droop Controller}
Droop control is a decentralized proportional controller with active and reactive power as control variables where the control gain (which is the droop gain) dictates the steady–state power distribution in the network. The \textit{conventional} droop controls (i.e. $P$-$f$/$Q$-$V$ droop) \cite{chandorkar} are derived for the off-grid microgrids under the assumption of a dominantly inductive network, i.e. for power lines with small R/X ratios. As discussed in Section~\ref{motivation}, this paper proposes a mode-dependent droop controller for $\mathrm{i}^{th}$ VSI for both on-grid and off-grid operation of the CI as shown in Fig.~\ref{fig:droop}. The $P$-$f$ droop control is considered here as a proportional controller (with proportional coefficient as $n_\mathrm{i}$) with error signal $e_\mathrm{P,i}:=(1-\mathcal{I}_\mathrm{gs})P_\mathrm{ref,i}-P_\mathrm{i}$ where $P_\mathrm{i}$ is the control variable and $(1-\mathcal{I}_\mathrm{gs})P_\mathrm{ref,i}$ is the reference. Whereas, the $Q$-$V$ droop control is considered here as a proportional-integral controller (with proportional and integral coefficients as $m_\mathrm{i}$ and $m_\mathrm{int,i}$ respectively) with error signal $e_\mathrm{Q,i}:=(1-\mathcal{I}_\mathrm{gs})Q_\mathrm{ref,i}-Q_\mathrm{i}$ where $Q_\mathrm{i}$ is the control variable and $(1-\mathcal{I}_\mathrm{gs})Q_\mathrm{ref,i}$ is the reference. The additional integral action in $Q$-$V$ droop control is required for on-grid mode only which is ensured by additional multiplication of $\mathcal{I}_\mathrm{gs}$ with the integral part. As a result, the proposed droop law is as follows:
\begin{align}
    \omega_\mathrm{r,i} &= \omega_\mathrm{nom} - n_\mathrm{i}[P_\mathrm{i}-(1-\mathcal{I}_\mathrm{gs})P_\mathrm{ref,i}], \label{droop1}\\
    V_\mathrm{r,i} &= V_\mathrm{nom} - m_\mathrm{i}[Q_\mathrm{i}-(1-\mathcal{I}_\mathrm{gs})Q_\mathrm{ref,i}] - \mathcal{I}_\mathrm{gs}m_\mathrm{int,i}\psi^\mathrm{Q}_\mathrm{i},\label{droop2}\\
    \psi^\mathrm{Q}_\mathrm{i} &= \int [Q_\mathrm{i}-(1-\mathcal{I}_\mathrm{gs})Q_\mathrm{ref,i}]\mathrm{d}t,\label{droop3}
\end{align}
where, $\omega_\mathrm{nom}$, $V_\mathrm{nom}$ are the nominal frequency (in rad/s) and voltage set-point (in volt) of the system respectively. $P_\mathrm{ref,i}$ and $Q_\mathrm{ref,i}$ are the active and reactive power set points, which are commonly set to active and reactive power rating of the $\mathrm{i}^{th}$ VSI respectively. The proposed droop control law differs from conventional droop characteristics in the following way:
\begin{enumerate}
    \item Unlike \textit{conventional} droop control law, an additional integral action, as defined in \eqref{droop3}, is introduced in $Q$-$V$ droop equation. This results in a proportional controller for active power and proportional-integral controller for reactive power with $(1-\mathcal{I}_\mathrm{gs})P_\mathrm{ref,i}$ and $(1-\mathcal{I}_\mathrm{gs})Q_\mathrm{ref,i}$ as reference signals respectively.
    \item Status of the feeder switch of system in Fig.~\ref{fig:system1}, $\mathcal{I}_\mathrm{gs}$, is included in the droop equation that makes the droop law strategy mode dependent (on-grid/off-grid mode). The function of $\mathcal{I}_\mathrm{gs}$ is to modify the droop law based on the transition from on-grid ($\mathcal{I}_\mathrm{gs}=1$) to off-grid mode ($\mathcal{I}_\mathrm{gs}=0$).
\end{enumerate}
The proposed controller architecture facilitates  the following essential functionality of VSIs for CIs as mentioned below:
\begin{enumerate}
\item The addition of the integral term, $\psi^\mathrm{Q}_\mathrm{i}$, facilitates the VSIs in on-grid mode to supply no reactive power in steady-state. This is further emphasized in Section~\ref{analysis}.A.
\item The addition of dependency on the variable, $\mathcal{I}_\mathrm{gs}$, facilitates the VSIs seamless functionality for CIs during the transition of on-/off-grid and off-/on-grid modes.
\end{enumerate}
\par A supervisory remote island detection algorithm is fast and accurate enough which, by means of any low-bandwidth communication channel, can convey the status, $\mathcal{I}_\mathrm{gs}$, from SEPSS of CIs to its VSIs. This in turn will allow the switching of droop functionality when a grid disconnection event occurs upon grid failure. The values of $n_\mathrm{i}$ and $m_\mathrm{i}$ are typically chosen such that $\omega_\mathrm{r,i}$ and $V_\mathrm{r,i}$ are within the allowed regulation, defined by IEEE $1547$ Standard \cite{ieee1547}, for all $P_\mathrm{i} \in [0, P_\mathrm{rated,i}]$ and $Q_\mathrm{i} \in [-Q_\mathrm{rated,i}, Q_\mathrm{rated,i}]$ respectively \cite{sanjay}. $P_\mathrm{rated,i}$ and $Q_\mathrm{rated,i}$ are the rated active and reactive powers that can be delivered by the $\mathrm{i}^{th}$ inverter. Although, these empirical upper bounds of droop co-efficient facilitate the initial design of droop law, system stability-constraint bounds of $n_\mathrm{i}$, $m_\mathrm{i}$ and $m_\mathrm{int,i}$ require special attention due to the system interconnection and its seamless transition between on-grid and off-grid mode of operation. Therefore, a more mathematically rigorous modeling and analysis of the system is required in order to assess the system stability and the power quality during transitions, and the relation of the performance goals with the control parameters. This will be discussed in Section~\ref{analysis}.
\begin{figure}[t]
	\centering
    \includegraphics[scale=0.26,trim={0cm 9.5cm 0cm 0cm},clip]{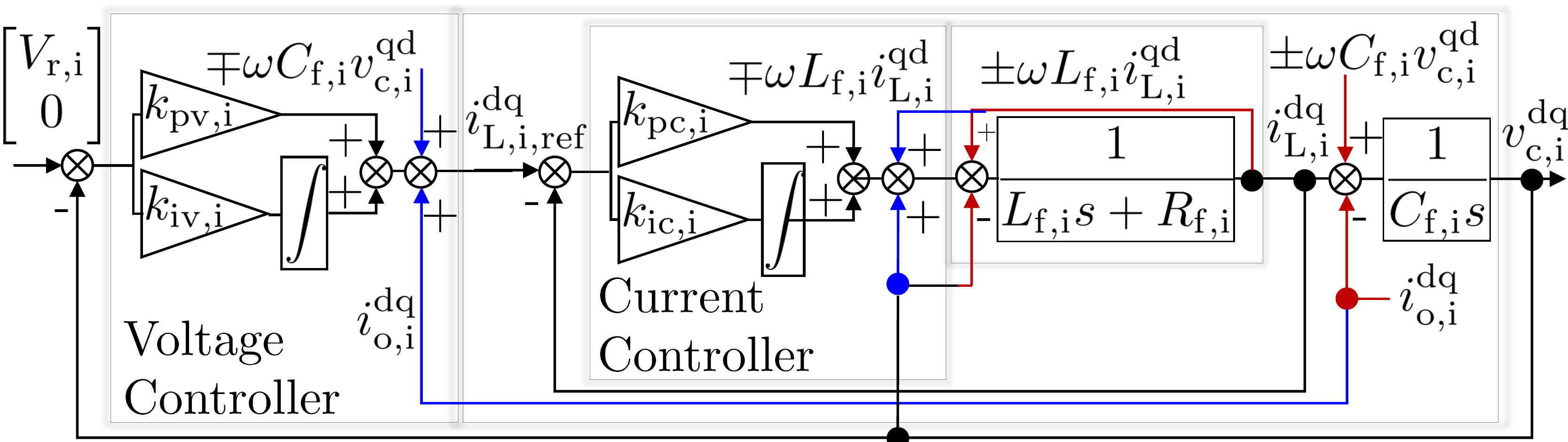}%
	\caption{Inner-current-outer-voltage control loop for $\mathrm{dq}$-axis control.}
	\label{fig:controller}
\end{figure}
\subsubsection{Inner-current-outer-voltage Controller}
The inner-current-outer-voltage controller architecture is employed, as illustrated in Fig.~\ref{fig:controller}, for the $3$-$\phi$ VSIs interconnection system \cite{iravani}. For the inner-current controller $i^\mathrm{dq}_\mathrm{L,i,ref}$ is provided as the reference signal to be tracked by the output signal, $i^\mathrm{dq}_\mathrm{L,i}$. A proportional-integral (PI) compensator is used for tracking of the reference of the $\mathrm{dq}$-axis inductor current for the $1^{st}$ order plant model as shown in Fig.~\ref{fig:controller}. For a desired time constant, $\tau_\mathrm{c,i}$, the parameters of the current controller are selected as $k_\mathrm{pc,i}=L_\mathrm{f,i}/\tau_\mathrm{c,i}$ and $k_\mathrm{ic,i}=R_\mathrm{f,i}/\tau_\mathrm{c,i}$. Depending on the switching frequency, $\tau_\mathrm{c,i}$ is typically selected to be in the range of $0.5$-$2$~ms \cite{iravani}. Additional feed-forward signals, $v^\mathrm{dq}_\mathrm{c,i}$ and $\mp \omega L_\mathrm{f,i}i^\mathrm{qd}_\mathrm{L,i}$ facilitate the disturbance rejection capability. For outer-voltage controller,  $\begin{bmatrix}V_\mathrm{r,i} & 0\end{bmatrix}^\top$ is considered to be the reference signal to be tracked by the VSI output voltage signal, $v^\mathrm{dq}_\mathrm{c,i}$. A PI compensator is used to enable reference tracking for the second order plant model with poles at $s=0, -1/\tau_\mathrm{c,i}$. For a desired phase margin and gain cross-over frequency, the parameters of the voltage controller ($k_\mathrm{pv,i}$ and $k_\mathrm{iv,i}$) can be designed based on \textit{symmetrical optimum} method \cite{iravani}. Similarly, additional feed-forward signals, $v^\mathrm{dq}_\mathrm{o,i}$ and $\mp \omega C_\mathrm{f,i}v^\mathrm{qd}_\mathrm{c,i}$ facilitate the disturbance rejection capability for the outer voltage control loop.
\section{Modeling of System Components}\label{model}
The system of Fig.~\ref{fig:system1} is considered in this section for modeling and analysing. Both the VSIs are employed with above mentioned controllers and connected to the PCC via $\mathrm{RL}$ coupling lines with $L_\mathrm{line,i}$, $R_\mathrm{line,i}$ for $\mathrm{i}$=$1,2$. By defining $L_\mathrm{l,i}:=L_\mathrm{g,i}+L_\mathrm{line,i}$ and $R_\mathrm{l,i}:=R_\mathrm{g,i}+R_\mathrm{line,i}$ for $\mathrm{i}=1,2$ (refer to Fig.~\ref{fig:circuit}), model of individual VSI will be developed in this section. The grid is assumed to be stiff, in order to emulate a real-world power system network. The grid impedance and feeder transformer leakage impedance (of Fig.~\ref{fig:system1}) are lumped and denoted by $L_\mathrm{l,g}$ and $R_\mathrm{l,g}$. Moreover, the complete R-L building network of CI along with loads are represented by an equivalent lumped load ($R_\mathrm{L}$ and $L_\mathrm{L}$) terminated at PCC. In this section, the following practical assumptions are made for further modeling the individual components of the system.
\begin{assumption}\label{assumption1}
The inner-current-outer-voltage control loop is stable and has faster dynamics when compared to the outer-most power and droop controller \cite{reduce1,iravani}.
\end{assumption}
\begin{assumption}\label{assumption2}
The outer-voltage controller tracks its voltage reference with minimal tracking error \cite{iravani}.
\end{assumption}
\begin{assumption}\label{assumption3}
Grid is stiff with nominal voltage ($V_\mathrm{nom}$) and frequency ($\omega_\mathrm{nom}$). 
\end{assumption}
\begin{assumption}\label{assumption4}
Load is assumed to be a constant $\mathrm{Z}$-type \cite{kundur}.
\end{assumption}
\begin{assumption}\label{assumption5}
All $\mathrm{abc}$-$\mathrm{dq}$ conversions are adopted w.r.t. a common reference with phase angle, $\theta_\mathrm{ref}$, and frequency, $\omega_\mathrm{ref}$.
\end{assumption}
\begin{remark}\label{remark1}
Assumption~\ref{assumption1} and \ref{assumption2} allow to model each VSI as AC voltage source with controllable phase, frequency and amplitude, bypassing all the internal states.
\end{remark}
These assumptions are made to assess the stability of the system and to obtain guidance on the selection of parameters. Efficacy of the control strategy is further corroborated by realistic CHIL simulations where these assumptions do not necessarily hold.
\begin{figure}[t]
	\centering
	\subfloat[]{\includegraphics[scale=0.27,trim={0cm 0cm 22cm 11cm},clip]{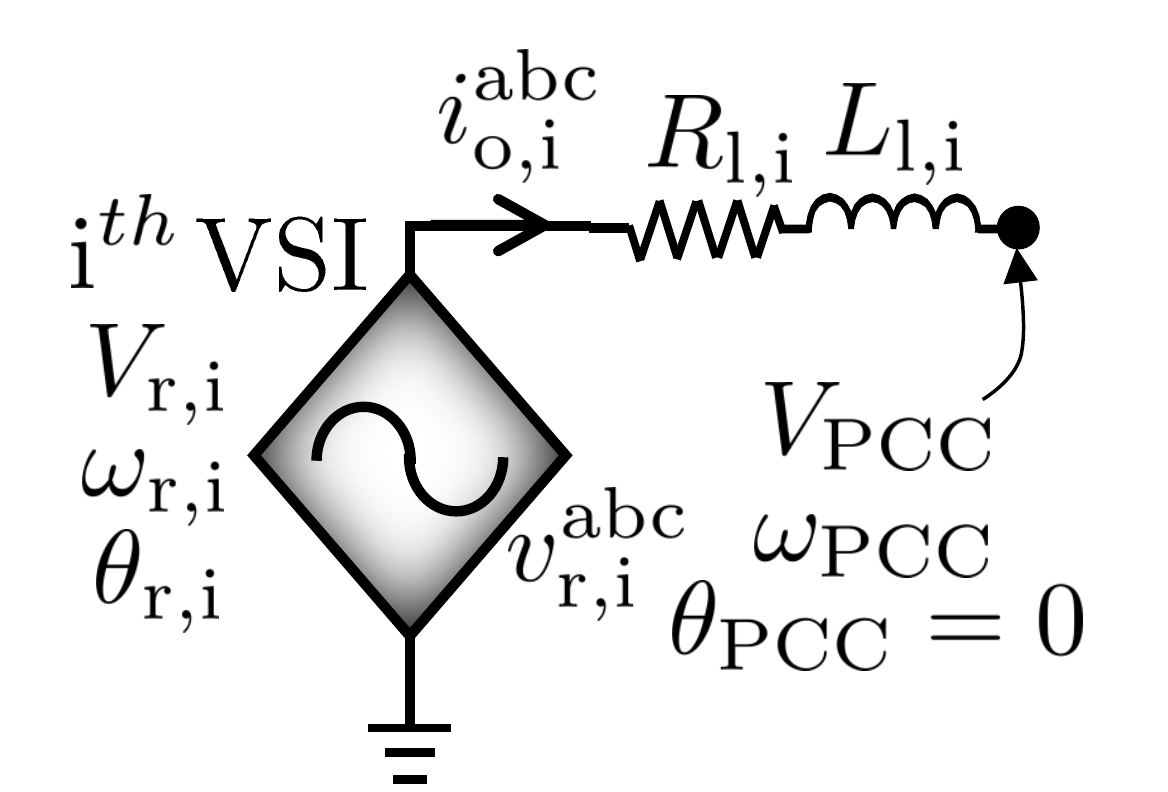}%
		\label{fig:comp1}}~
	\subfloat[]{\includegraphics[scale=0.29,trim={0cm 0cm 25cm 11cm},clip]{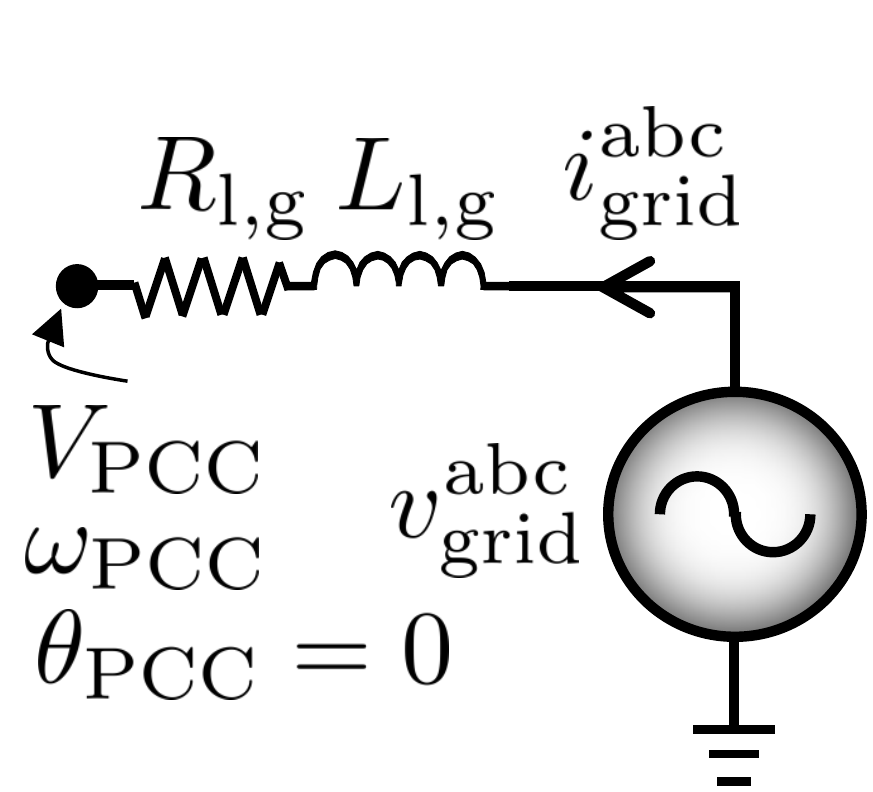}%
		\label{fig:comp2}}~
	\subfloat[]{\includegraphics[scale=0.28,trim={0cm 0cm 22cm 10.5cm},clip]{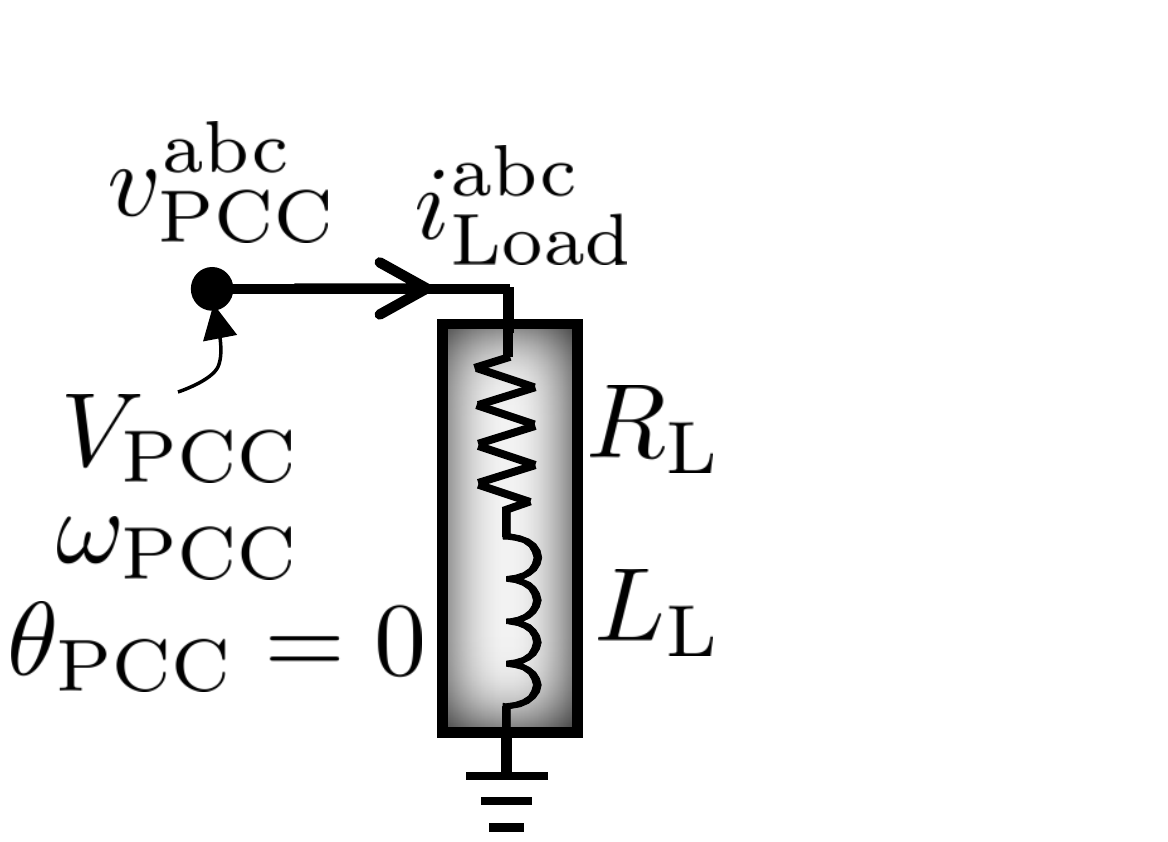}%
		\label{fig:comp3}}
	\caption{Modeling of individual components of the system shown in Fig.~\ref{fig:system1}, (a) model of $\mathrm{i}^{th}$ VSI, (b) model of grid, (c) R-L model of equivalent network and load, connected to the PCC}
	\label{fig:comp}
\end{figure}
\subsection{Modelling of a Single VSI in Off-grid Mode ($\mathcal{I}_\mathrm{gs}$=$0$)}
Noting Remark~\ref{remark1} and following the model of Fig.~\ref{fig:comp}\subref{fig:comp1} in off-grid mode with $\mathrm{i}^{th}$ VSI, the following non-linear state-space equations can be obtained for $\mathrm{i}=1,2$:
\begin{align}\label{thetaID0}
    \dfrac{\mathrm{d\theta_\mathrm{i}}}{\mathrm{d}t} = \omega_\mathrm{r,i}-\omega_\mathrm{ref},~\text{where}~\theta_\mathrm{i}:=\theta_\mathrm{r,i}-\theta_\mathrm{ref}.
\end{align}
$\theta_\mathrm{r,i}$ is the phase angle of $\mathrm{i}^{th}$ VSI w.r.t. the PCC ($\theta_\mathrm{PCC}=0$). Using \eqref{powercal}, Assumption~\ref{assumption5} and \eqref{droop1} with $\mathcal{I}_\mathrm{gs}=0$ for off-grid mode, the frequency dynamics can be written as:
\begin{align}
    \tau_\mathrm{S,i}\dfrac{\mathrm{d\omega_\mathrm{r,i}}}{\mathrm{d}t}&=\omega_\mathrm{nom}-\omega_\mathrm{r,i}-n_\mathrm{i}[p_\mathrm{i} -P_\mathrm{ref,i}],\label{freqID0a}\\
    \text{where,~}p_\mathrm{i}&:=\dfrac{3}{2}[V_\mathrm{r,i}i_\mathrm{o,i}^\mathrm{d}\cos\theta_\mathrm{i}+V_\mathrm{r,i}i_\mathrm{o,i}^\mathrm{q}\sin\theta_\mathrm{i}].\label{freqID0b}
\end{align}
Similarly, using \eqref{powercal}, Assumption~\ref{assumption5} and \eqref{droop2} with $\mathcal{I}_\mathrm{gs}=0$ (off-grid mode), voltage dynamics can be written as:
\begin{align}
    \tau_\mathrm{S,i}\dfrac{\mathrm{d}V_\mathrm{r,i}}{\mathrm{d}t}&=V_\mathrm{nom}-V_\mathrm{r,i}-m_\mathrm{i}[q_\mathrm{i} -Q_\mathrm{ref,i}],\label{voltID0a}\\
    \text{where,~}q_\mathrm{i}&:=\dfrac{3}{2}[-V_\mathrm{r,i}i_\mathrm{o,i}^\mathrm{q}\cos\theta_\mathrm{i}+V_\mathrm{r,i}i_\mathrm{o,i}^\mathrm{d}\sin\theta_\mathrm{i}].\label{voltID0b}
\end{align}
In addition, the coupling lines impose the follow dynamics using Assumption~\ref{assumption5}:
\begin{align}
    L_\mathrm{l,i}\dfrac{\mathrm{d}i_\mathrm{o,i}^\mathrm{d}}{\mathrm{d}t}&=V_\mathrm{r,i}\cos\theta_\mathrm{i}-R_\mathrm{l,i}i_\mathrm{o,i}^\mathrm{d}-v_\mathrm{PCC}^\mathrm{d}+\omega_{\mathrm{ref}}L_\mathrm{l,i}i_\mathrm{o,i}^\mathrm{q},\label{ida}\\
    L_\mathrm{l,i}\dfrac{\mathrm{d}i_\mathrm{o,i}^\mathrm{q}}{\mathrm{d}t}&=V_\mathrm{r,i}\sin\theta_\mathrm{i}-R_\mathrm{l,i}i_\mathrm{o,i}^\mathrm{q}-v_\mathrm{PCC}^\mathrm{q}-\omega_{\mathrm{ref}}L_\mathrm{l,i}i_\mathrm{o,i}^\mathrm{d},\label{idb}
\end{align}
where, $v_\mathrm{PCC}^\mathrm{d}:=V_\mathrm{PCC}\cos \theta_\mathrm{ref}$ and $v_\mathrm{PCC}^\mathrm{q}:=-V_\mathrm{PCC}\sin \theta_\mathrm{ref}$. This results in a $5^{th}$-order non-linear electromagnetic (EM) model for the droop-controlled VSI operating in off-grid mode.
\subsection{Modelling of a Single VSI in on-grid Mode ($\mathcal{I}_\mathrm{gs}$=$1$)}
Noting Remark~\ref{remark1} and following the model of Fig.~\ref{fig:comp}\subref{fig:comp1} in on-grid mode with $\mathrm{i}^{th}$ VSI, the angular state-space equation remains the same as \eqref{thetaID0}.
Again using \eqref{powercal}, Assumption~\ref{assumption5} and \eqref{droop1} with $\mathcal{I}_\mathrm{gs}=1$ for on-grid mode, frequency dynamics can be written as:
\begin{align}
    \tau_\mathrm{S,i}\dfrac{\mathrm{d\omega_\mathrm{r,i}}}{\mathrm{d}t}&=\omega_\mathrm{nom}-\omega_\mathrm{r,i}-n_\mathrm{i}p_\mathrm{i},\label{freqID1a}\\
    \text{where,~}p_\mathrm{i}&:=\dfrac{3}{2}[V_\mathrm{r,i}i_\mathrm{o,i}^\mathrm{d}\cos\theta_\mathrm{i}+V_\mathrm{r,i}i_\mathrm{o,i}^\mathrm{q}\sin\theta_\mathrm{i}].\label{freqID1b}
\end{align}
Using \eqref{droop2} and \eqref{droop3} with $\mathcal{I}_\mathrm{gs}=1$, the following can be derived as:
\begin{align}
    \dfrac{\mathrm{d}\psi^\mathrm{Q}_\mathrm{i}}{\mathrm{d}t}&=Q_\mathrm{i} = (m_\mathrm{i})^{-1}[V_\mathrm{nom}-V_\mathrm{r,i}-m_\mathrm{int,i}\psi^\mathrm{Q}_\mathrm{i}] \nonumber\\
    \implies m_\mathrm{i}\dfrac{\mathrm{d}\psi^\mathrm{Q}_\mathrm{i}}{\mathrm{d}t} &= V_\mathrm{nom}-V_\mathrm{r,i}-m_\mathrm{int,i}\psi^\mathrm{Q}_\mathrm{i}.\label{psi}
\end{align}
Using Assumption~\ref{assumption5}, \eqref{powercal} and \eqref{droop2} with $\mathcal{I}_\mathrm{gs}=1$ and \eqref{psi}, voltage dynamics can be written as:
\begin{align}
    \tau_\mathrm{S,i}\dfrac{\mathrm{d}V_\mathrm{r,i}}{\mathrm{d}t}&=k_\mathrm{m,i}(V_\mathrm{nom}-V_\mathrm{r,i}-m_\mathrm{int,i}\psi^\mathrm{Q}_\mathrm{i})-m_\mathrm{i}q_\mathrm{i},\label{voltID1a}\\
    q_\mathrm{i}&:=\dfrac{3}{2}[-V_\mathrm{r,i}i_\mathrm{o,i}^\mathrm{q}\cos\theta_\mathrm{i}+V_\mathrm{r,i}i_\mathrm{o,i}^\mathrm{d}\sin\theta_\mathrm{i}],\label{voltID1b}
\end{align}
where, $k_\mathrm{m,i}:=1- \tau_\mathrm{S,i}m_\mathrm{int,i}/m_\mathrm{i}$. Similarly, the coupling line imposes dynamics same as given in \eqref{ida} and \eqref{idb}.
This results in a $6^{th}$-order non-linear EM model for the droop-controlled VSI operating in on-grid mode.
\subsection{Modelling of the Stiff Grid}
Following the model as shown in Fig.~\ref{fig:comp}\subref{fig:comp2} and Assumption~\ref{assumption3}, the following non-linear state-space equations can be derived for the grid
\begin{align}\label{thetagrid}
    \dfrac{\mathrm{d\theta_\mathrm{g}}}{\mathrm{d}t} &= \omega_\mathrm{grid}-\omega_\mathrm{ref},~\text{where}~\theta_\mathrm{g}:=\theta_\mathrm{grid}-\theta_\mathrm{ref},\\
    \omega_\mathrm{grid}&=\omega_\mathrm{nom},~\text{and},~V_\mathrm{g}=V_\mathrm{nom}.\nonumber
\end{align}
$\theta_\mathrm{grid}$ is the phase angle of the grid w.r.t. the PCC. Similarly, the line dynamics in Fig.~\ref{fig:comp}\subref{fig:comp2} add two more states, with Assumption~\ref{assumption5}, as follows:
\begin{align}
    L_\mathrm{l,g}\dfrac{\mathrm{d}i_\mathrm{g}^\mathrm{d}}{\mathrm{d}t}&=V_\mathrm{g}\cos\theta_\mathrm{g}-R_\mathrm{l,g}i_\mathrm{g}^\mathrm{d}-v_\mathrm{PCC}^\mathrm{d}+\omega_{\mathrm{ref}}L_\mathrm{l,g}i_\mathrm{g}^\mathrm{q},\label{idg}\\
    L_\mathrm{l,g}\dfrac{\mathrm{d}i_\mathrm{g}^\mathrm{q}}{\mathrm{d}t}&=V_\mathrm{g}\sin\theta_\mathrm{g}-R_\mathrm{l,g}i_\mathrm{g}^\mathrm{q}-v_\mathrm{PCC}^\mathrm{q}-\omega_{\mathrm{PCC}}L_\mathrm{l,g}i_\mathrm{g}^\mathrm{d}.\label{iqg}
\end{align}
This results in a $3^{rd}$-order non-linear EM model for the stiff grid at nominal voltage and frequency.
\subsection{Modelling of Load at PCC}
Using Assumption~\ref{assumption4} and \ref{assumption5}, the following dynamics equations can be derived for Fig.~\ref{fig:comp}\subref{fig:comp3}:
\begin{align}
    L_\mathrm{L}\dfrac{\mathrm{d}i_\mathrm{Load}^\mathrm{d}}{\mathrm{d}t} &= v_\mathrm{PCC}^\mathrm{d}-R_\mathrm{L}i_\mathrm{Load}^\mathrm{d} + \omega_\mathrm{ref}L_\mathrm{L}i_\mathrm{Load}^\mathrm{q}, \label{load1}\\
    L_\mathrm{L}\dfrac{\mathrm{d}i_\mathrm{Load}^\mathrm{q}}{\mathrm{d}t} &= v_\mathrm{PCC}^\mathrm{q}-R_\mathrm{L}i_\mathrm{Load}^\mathrm{q} - \omega_\mathrm{ref}L_\mathrm{L}i_\mathrm{Load}^\mathrm{d}. \label{load2}
\end{align}
Depending on on-grid/off-grid mode, $i_\mathrm{Load}^\mathrm{d}$ and $i_\mathrm{Load}^\mathrm{q}$ can be replaced as $i_\mathrm{Load}^\mathrm{d}=\mathcal{I}_\mathrm{gs}~i_\mathrm{g}^\mathrm{d}+\sum\limits_{i=1,2}i_\mathrm{o,i}^\mathrm{d}$ and $i_\mathrm{Load}^\mathrm{q}=\mathcal{I}_\mathrm{gs}~i_\mathrm{g}^\mathrm{q}+\sum\limits_{i=1,2}i_\mathrm{o,i}^\mathrm{q}$. By using \eqref{ida}, \eqref{idb}, \eqref{idg} and \eqref{iqg} in \eqref{load1} and \eqref{load2}, the following algebraic equation can be obtained to model the load at PCC:
\begin{align}\label{load}
    &v_\mathrm{PCC}^\mathrm{d}=\mathcal{I}_\mathrm{gs}[k^\mathrm{V}_\mathrm{g}V_\mathrm{g}\cos\theta_\mathrm{g}+k^\mathrm{C}_\mathrm{g}i_\mathrm{g}^\mathrm{d}]+\hspace*{-2mm}\sum_{i=1,2}[k^\mathrm{V}_\mathrm{i}V_\mathrm{r,i}\cos\theta_\mathrm{i}+k^\mathrm{C}_\mathrm{i}i_\mathrm{o,i}^\mathrm{d}], \\
    &v_\mathrm{PCC}^\mathrm{q}=\mathcal{I}_\mathrm{gs}[k^\mathrm{V}_\mathrm{g}V_\mathrm{g}\sin\theta_\mathrm{g}+k^\mathrm{C}_\mathrm{g}i_\mathrm{g}^\mathrm{q}]+\hspace*{-2mm}\sum_{i=1,2}[k^\mathrm{V}_\mathrm{i}V_\mathrm{r,i}\sin\theta_\mathrm{i}+k^\mathrm{C}_\mathrm{i}i_\mathrm{o,i}^\mathrm{q}], \\
    &k^\mathrm{V}_\mathrm{g}=\frac{\frac{1}{L_\mathrm{l,g}}}{k^\mathrm{T}},~k^\mathrm{V}_\mathrm{i}=\frac{\frac{1}{L_\mathrm{l,i}}}{k^\mathrm{T}},~k^\mathrm{C}_\mathrm{g}=\frac{\frac{R_\mathrm{L}}{L_\mathrm{L}}-\frac{R_\mathrm{l,g}}{L_\mathrm{l,g}}}{k^\mathrm{T}},~k^\mathrm{C}_\mathrm{i}=\frac{\frac{R_\mathrm{L}}{L_\mathrm{L}}-\frac{R_\mathrm{l,i}}{L_\mathrm{l,i}}}{k^\mathrm{T}}, \nonumber
\end{align}
where, $k^\mathrm{T}=1/L_\mathrm{L}+\mathcal{I}_\mathrm{gs}/L_\mathrm{g}+\sum\limits_{i=1,2}1/L_\mathrm{l,i}$.
As a result, the complete system model of Fig.~\ref{fig:system1} along with the control loops can be equivalently represented as shown in Fig.~\ref{fig:system}. In the next section, the combined state space model of the system shown in Fig.~\ref{fig:system}, will be developed and the stability will be assessed for different operating conditions.
\section{Analysis of System Model and Stability}\label{analysis}
\begin{figure}[t]
	\centering
    \includegraphics[scale=0.26,trim={0cm 0cm 0cm 0cm},clip]{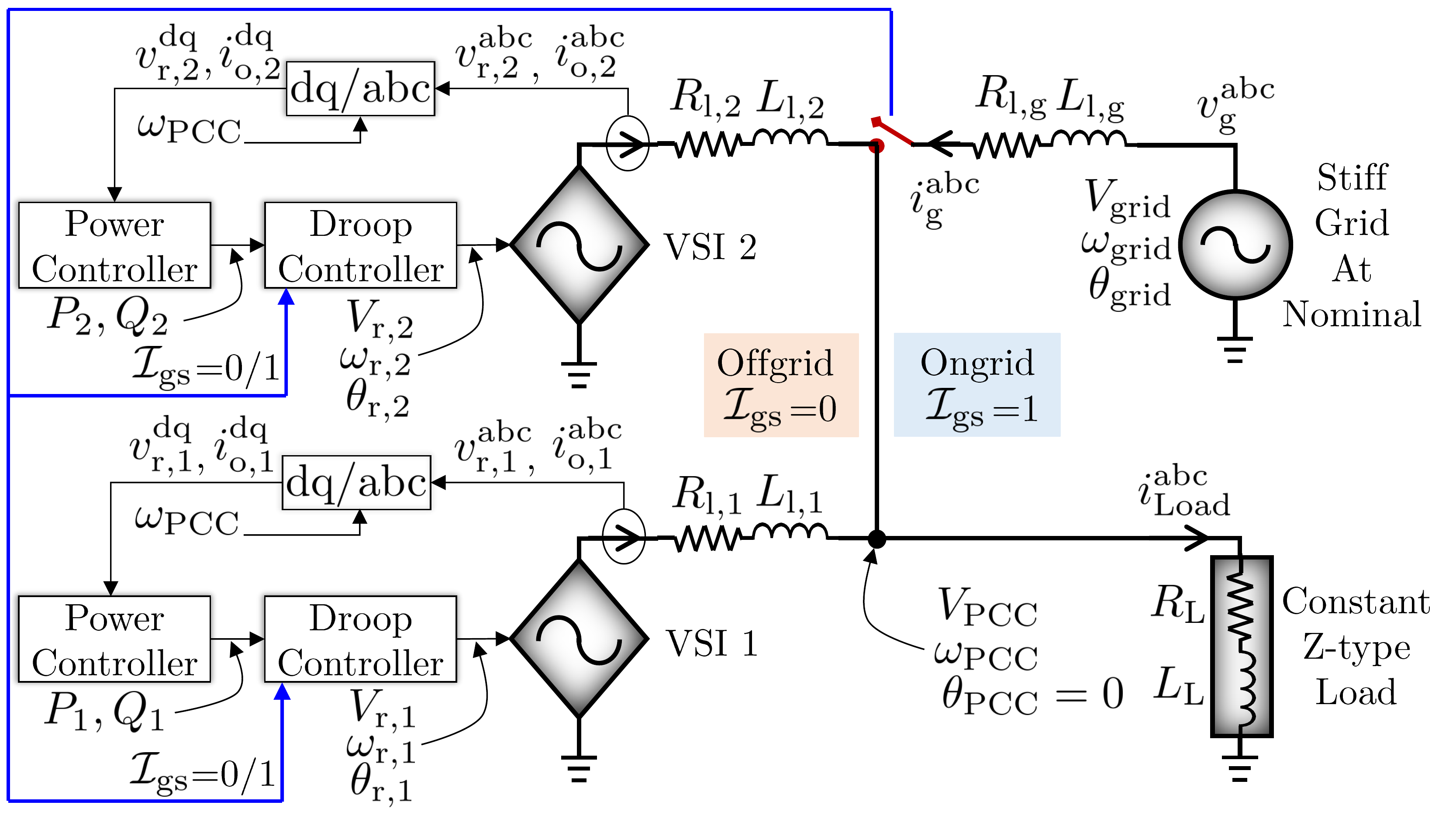}%
	\caption{Equivalent system diagram of Fig.~\ref{fig:system1}.}
	\label{fig:system}
\end{figure}
\subsection{System operating in on-grid mode}
While the system of Fig.~\ref{fig:system} is operating in on-grid mode, i.e. $\mathcal{I}_\mathrm{gs}=1$, it can be modelled as the following non-linear system:
\begin{align}\label{ongridmodel}
    \dfrac{\mathrm{d}\underline{x}_\mathrm{on}}{\mathrm{d}t}&=\underline{\mathcal{G}}_\mathrm{on}(\underline{x}_\mathrm{on}),
\end{align}
where, $\underline{\mathcal{G}}_\mathrm{on}(.)$ consists of \eqref{thetaID0}, \eqref{freqID1a}, \eqref{freqID1b}, \eqref{psi}, \eqref{voltID1a}, \eqref{voltID1b}, \eqref{ida}, \eqref{idb} for $\mathrm{i}=1,2$, \eqref{thetagrid}, \eqref{idg}, \eqref{iqg} and \eqref{load}. Here, $\underline{x}_\mathrm{on}^\top=[\theta_\mathrm{1}~\theta_\mathrm{2}~\theta_\mathrm{g}~\omega_\mathrm{r,1}~\omega_\mathrm{r,2}~V_\mathrm{r,1}~V_\mathrm{r,2}~\psi^\mathrm{Q}_\mathrm{1}~\psi^\mathrm{Q}_\mathrm{2}~i_\mathrm{o,1}^\mathrm{d}~i_\mathrm{o,2}^\mathrm{d}~i_\mathrm{g}^\mathrm{d}~i_\mathrm{o,1}^\mathrm{q}~i_\mathrm{o,2}^\mathrm{q}~i_\mathrm{g}^\mathrm{q}]$. As a result, the complete system in on-grid mode is a $15$-order nonlinear autonomous system with \textit{state} vector $\underline{x}_\mathrm{on}$.
\begin{theorem}\label{theorem1}
Consider the system as shown in Fig.~\ref{fig:system} holding Assumption~\ref{assumption3}, operating in on-grid mode where both the VSIs are following droop laws given in \eqref{droop1}, \eqref{droop2} and \eqref{droop3}. In steady-state, both the VSIs will supply zero active and reactive output power and the grid will supply the entire demand at PCC.
\end{theorem}
\begin{proof}
Proof can be found in the Appendix A.
\end{proof}
\noindent Based on the Theorem~\ref{theorem1}, following remarks can be made:
\begin{remark}\label{remark2}
In the proposed control for CI, battery-interfaced-VSIs with droop control law stated in \eqref{droop1}, \eqref{droop2} and \eqref{droop3} do not need to de-energize and can remain connected to the system, while the grid is available to supply all the electrical loads.
\end{remark}
\begin{remark}\label{remark3}
Seamless power flow can be achieved to CI while grid is disconnected as VSIs are always synchronized to the system and no transfer-switching process is required.
\end{remark}
However a systematic methodology, in order to characterize the ranges of various control parameters of the VSIs that lead to guaranteed stability of the system in on-grid mode needs to be characterized for designing the proposed controllers for CIs. Thus, the following problem can be formulated:
\begin{problem}\label{problem1}
Consider the system of \eqref{ongridmodel} linearized around
an equilibrium point, $\underline{x}_\mathrm{on}^\mathrm{eq}$
\begin{align}\label{linearongrid}
    \dfrac{\mathrm{d}\Delta \underline{x}_\mathrm{on}}{\mathrm{d}t}=\underbrace{\mathcal{\underline{F}}_\mathrm{on}(\underline{x}_\mathrm{on})\big|_\mathrm{\underline{x}_\mathrm{on}^\mathrm{eq}}}_{\mathbf{A}_\mathrm{on}}\Delta \underline{x}_\mathrm{on},
\end{align}
where $\mathbf{A_{\mathrm{on}}}$ is the Jacobian matrix of $\mathcal{\underline{F}}_\mathrm{on}(.)$, the vector field of \eqref{ongridmodel}. Develop a systematic framework for tuning the control parameters i.e. $n_\mathrm{1}$, $n_\mathrm{2}$, $m_\mathrm{1}$, $m_\mathrm{2}$, $m_\mathrm{int,1}$, $m_\mathrm{int,2}$, $\tau_\mathrm{1}$, $\tau_\mathrm{2}$, so that asymptotic stability of the equilibrium of \eqref{linearongrid} is guaranteed.
\end{problem}
The linearized model of the system is further reduced by suppressing the EM dynamics of the lines and loads. The assumption made here is that the characteristic EM timescale for lines and loads are slower ($\sim$~$3$-$5$~ms \cite{kundur}) and are below the characteristic timescale of droop controls ($\sim$~$30$-$40$~ms \cite{green}), dictated by the choice of the parameter, $\tau_\mathrm{S,i}$ \cite{dhople}. This strong time-scale separation in the system between the slow states denoted as $\underline{x}_\mathrm{on}^\mathrm{s}:=[\theta_\mathrm{1}~\theta_\mathrm{2}~\theta_\mathrm{g}~\omega_\mathrm{r,1}~\omega_\mathrm{r,2}~V_\mathrm{r,1}~V_\mathrm{r,2}~\psi^\mathrm{Q}_\mathrm{1}~\psi^\mathrm{Q}_\mathrm{2}]^\top$, and the fast states denoted as $\underline{x}_\mathrm{on}^\mathrm{f}:=[i_\mathrm{o,1}^\mathrm{d}~i_\mathrm{o,2}^\mathrm{d}~i_\mathrm{g}^\mathrm{d}~i_\mathrm{o,1}^\mathrm{q}~i_\mathrm{o,2}^\mathrm{q}~i_\mathrm{g}^\mathrm{q}]^\top$, allows the assumption that the fast states reach close to their quasi-steady-state values (i.e. $\underline{x}_\mathrm{on,ss}^\mathrm{f}$), much faster than the slow states \cite{modred1,modred2}. More elaborated discussion is presented in Appendix~B. Therefore, only remaining slow states are of interest from the stability point of view of the following reduced-order model of the system of \eqref{linearongrid} given as:
\begin{align}\label{redlinearongrid}
        \dfrac{\mathrm{d}\Delta \underline{x}_\mathrm{on}^\mathrm{s}}{\mathrm{d}t}=\mathbf{A}_\mathrm{{on}}^\mathrm{{s}}\Delta \underline{x}_\mathrm{on}^\mathrm{s}.
\end{align}
With the equilibrium point of the model in \eqref{redlinearongrid} and the parameters indicated in Table~\ref{table:control} and \ref{table:data}, the root locus-based analysis is conducted in order to obtain a range of various control parameters, as stated in Problem~\ref{problem1}, empirically. With the assumption of having homogeneous control parameters, i.e. $n_\mathrm{1}=n_\mathrm{2}$, $m_\mathrm{1}=m_\mathrm{2}$, $m_\mathrm{int,1}=m_\mathrm{int,2}$ and $\tau_\mathrm{1}=\tau_\mathrm{2}$, the complete eigenvalue spectrum of the system state matrix $\mathbf{A}^\mathrm{{s}}_\mathrm{{on}}$ are shown in Fig.~\ref{fig:locusgrid}\subref{locusngrid}, \ref{fig:locusgrid}\subref{locusmgrid}, \ref{fig:locusgrid}\subref{locustaugrid} and \ref{fig:locusgrid}\subref{locusmintgrid}. It is observed here that with the increase of these control parameters individually, some of the low-frequency eigen values gradually move away from the real axis, meanwhile being close to the imaginary axis until the eigenvalues reach unstable right-half zone. On the contrary, some of the stable eigenvalues gradually move on the real axis only away from the imaginary axis. It is evident that the low-frequency eigenvalues near the imaginary axis are the most dominant and crucial for system stability, while the eigenvalues close to the real axis will affect the dynamic performance and damping properties of the system \cite{muller}. An empirical range of various control parameters, with a numerical bounds, affecting the microgrid stability and dynamic performance are obtained from the critical value of the root locus as tabulated in Table~\ref{table:control}.
\begin{figure}[t]
	\centering
	\subfloat[]{\includegraphics[scale=0.25,trim={0cm 0cm 17cm 7.5cm},clip]{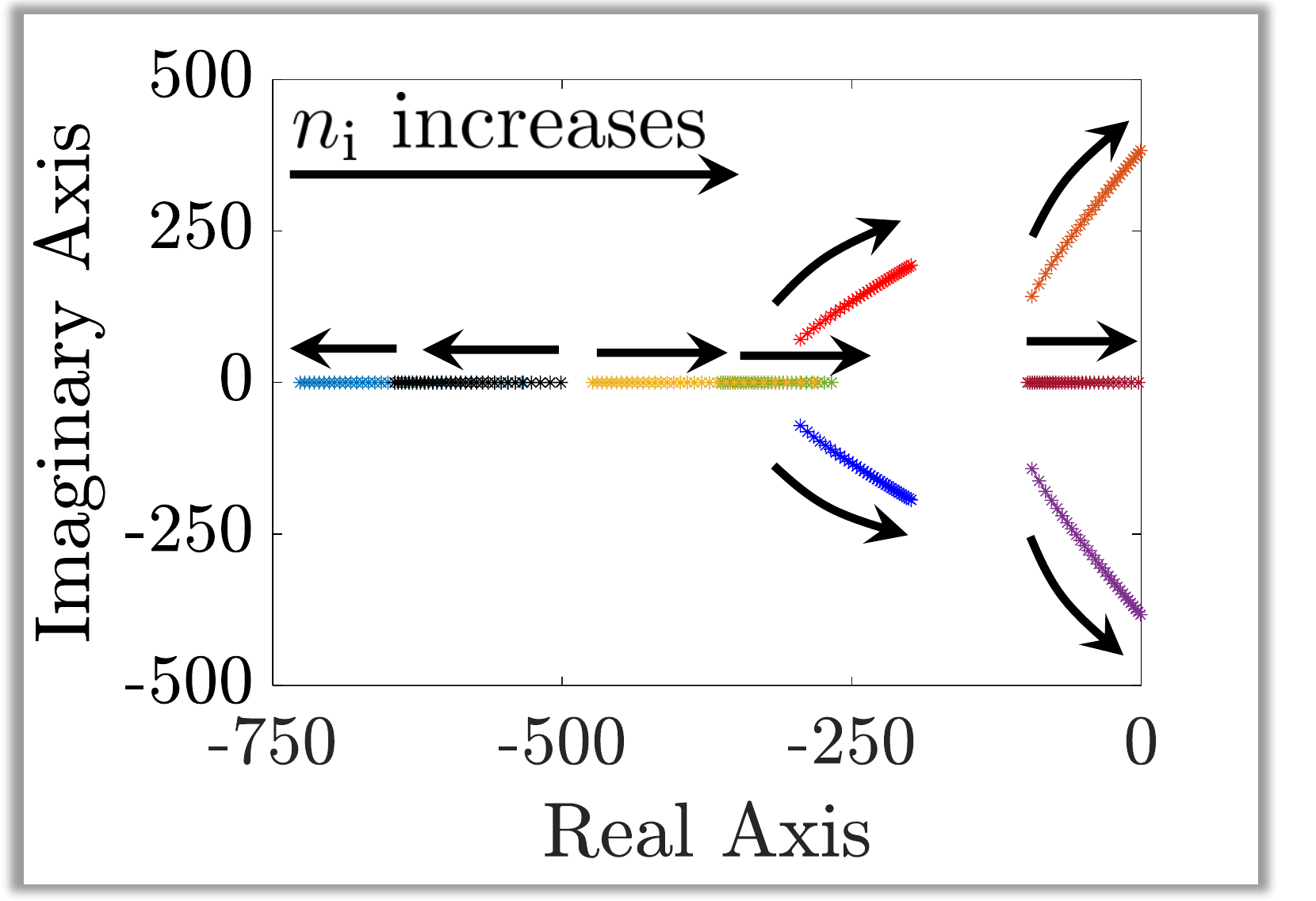}%
		\label{locusngrid}}~
	\subfloat[]{\includegraphics[scale=0.25,trim={0cm 0cm 17cm 7.5cm},clip]{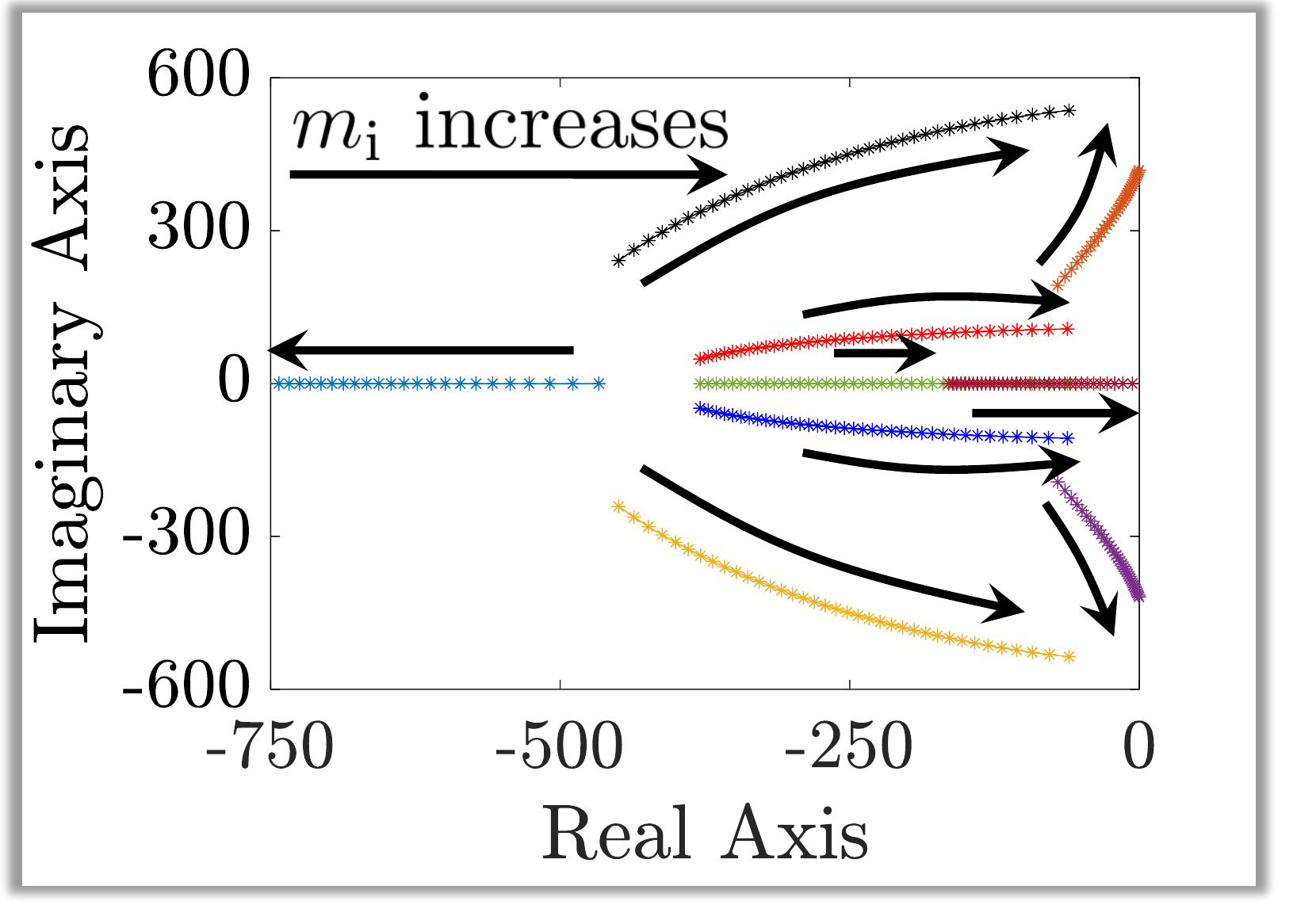}%
		\label{locusmgrid}}
		
	\subfloat[]{\includegraphics[scale=0.25,trim={0cm 0cm 17cm 7.5cm},clip]{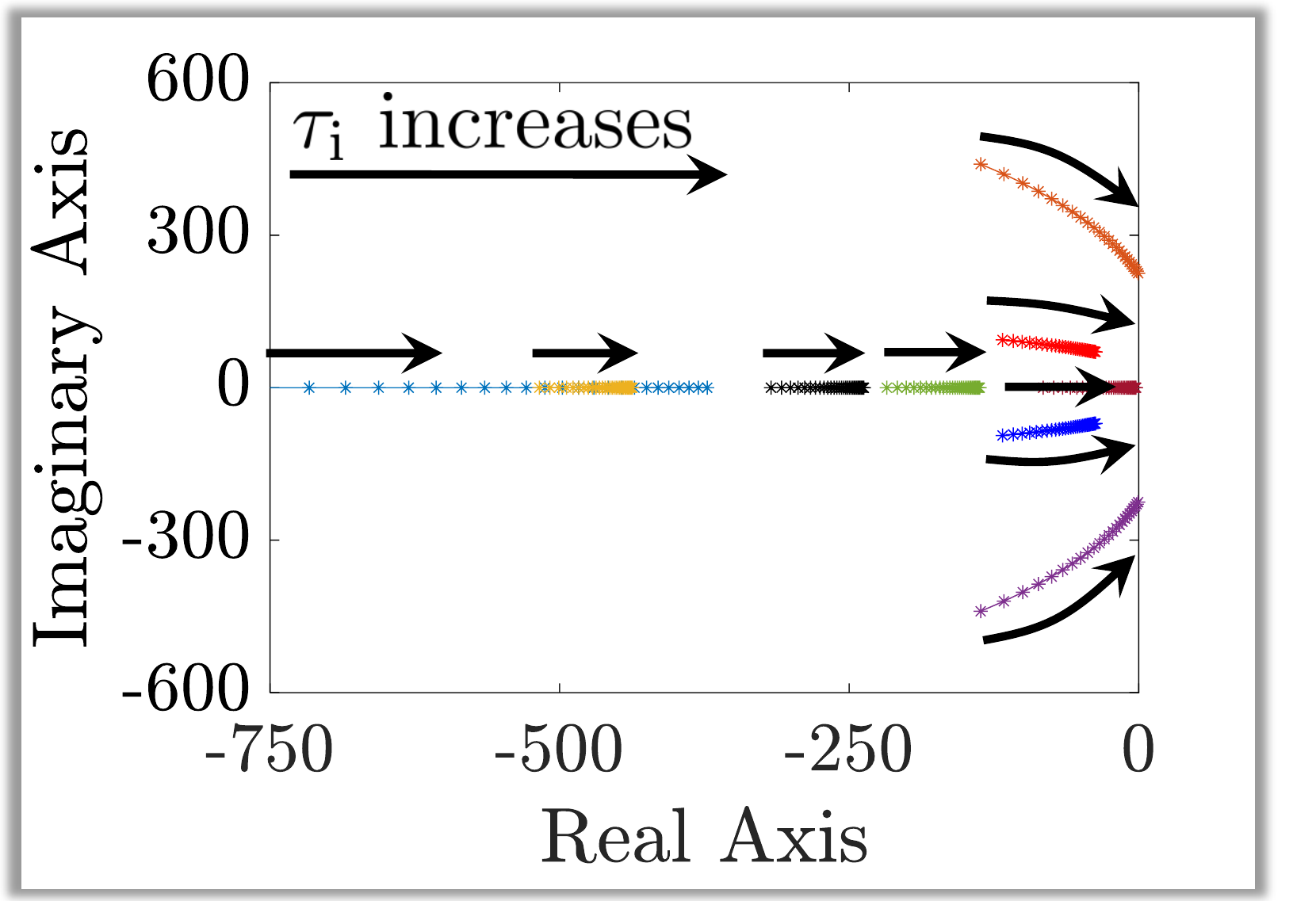}%
		\label{locustaugrid}}~
	\subfloat[]{\includegraphics[scale=0.25,trim={0cm 0cm 17cm 7.5cm},clip]{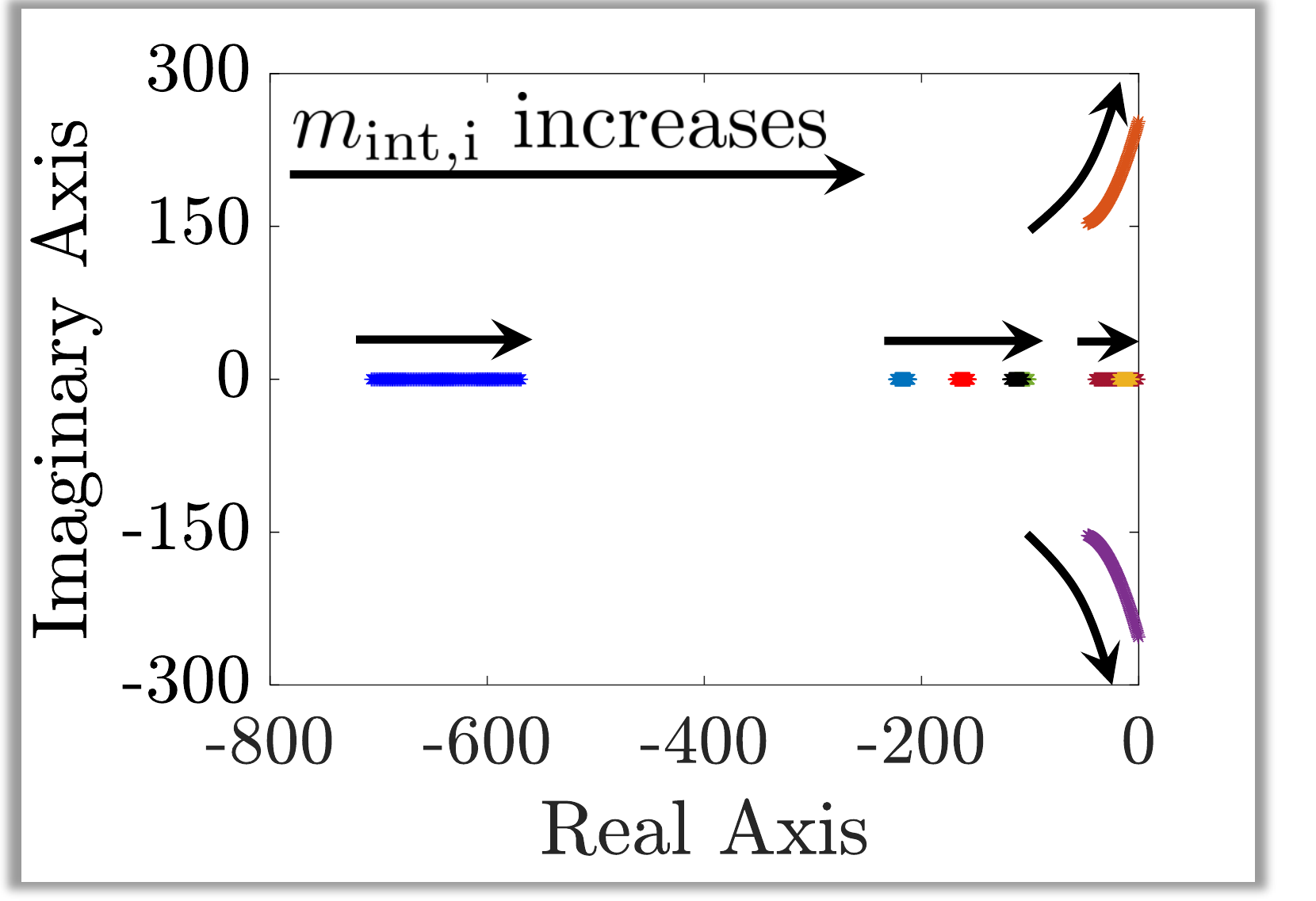}%
		\label{locusmintgrid}}
	\caption{Eigenvalue spectrum of the system in \eqref{redlinearongrid} with, (a) $1.04\times 10^{-2}\le n_\mathrm{i} \le 3.64\times 10^{-2}$~rad/s/kW, (b) $41.67\times 10^{-3} \le m_\mathrm{i} \le 416.6\times 10^{-3}$~volt/kVAr, (c) $24.7 \le \tau_\mathrm{S,i} \le 41.2$~ms, (d) $0.6 \le m_\mathrm{int,i} \le 0.84$~volt/s/kVAr.}
	\label{fig:locusgrid}
\end{figure}
\begin{figure}[t]
	\centering
	\subfloat[]{\includegraphics[scale=0.25,trim={0cm 0cm 17cm 7.5cm},clip]{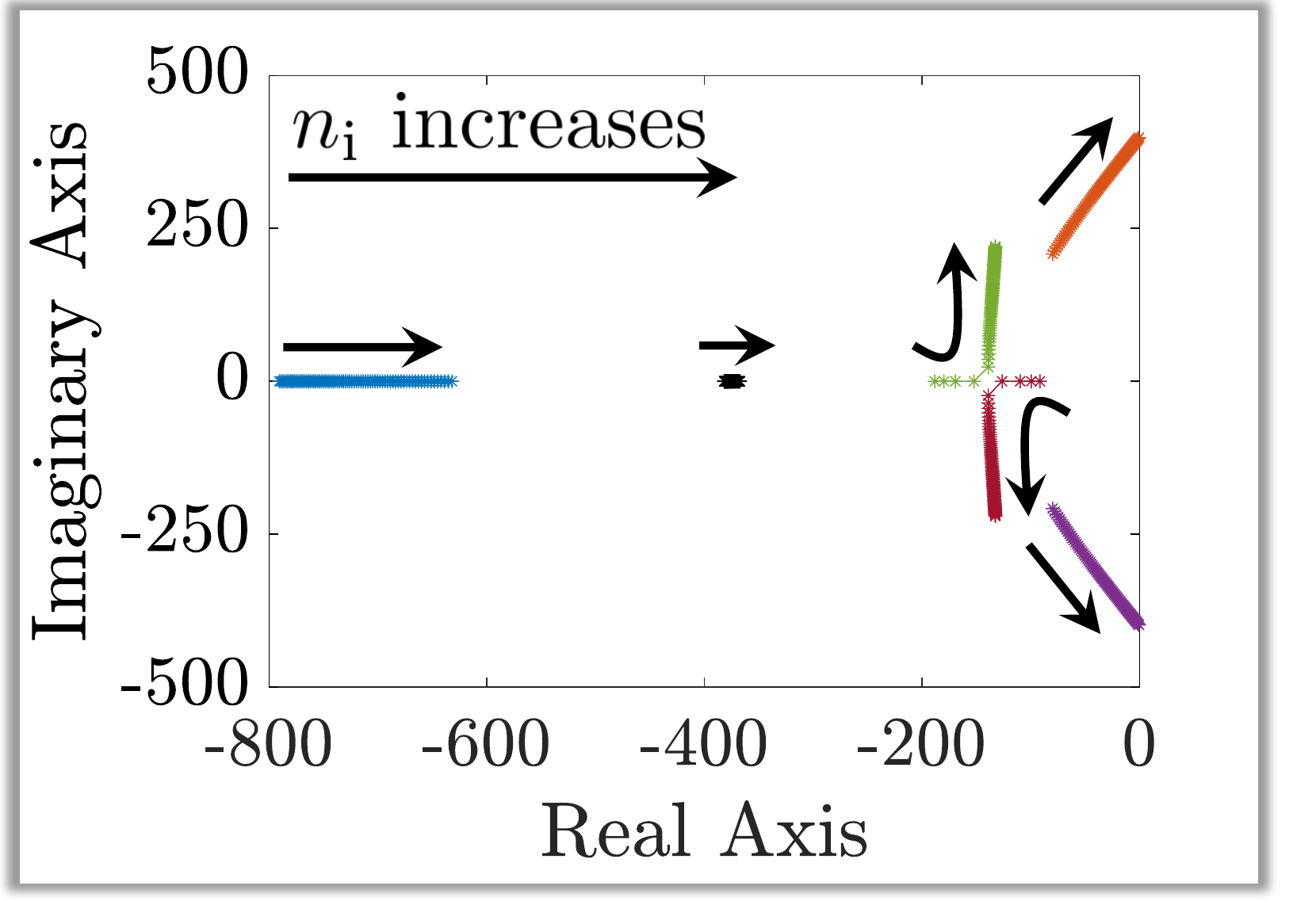}%
		\label{locusnisland}}~
	\subfloat[]{\includegraphics[scale=0.25,trim={0cm 0cm 17cm 7.5cm},clip]{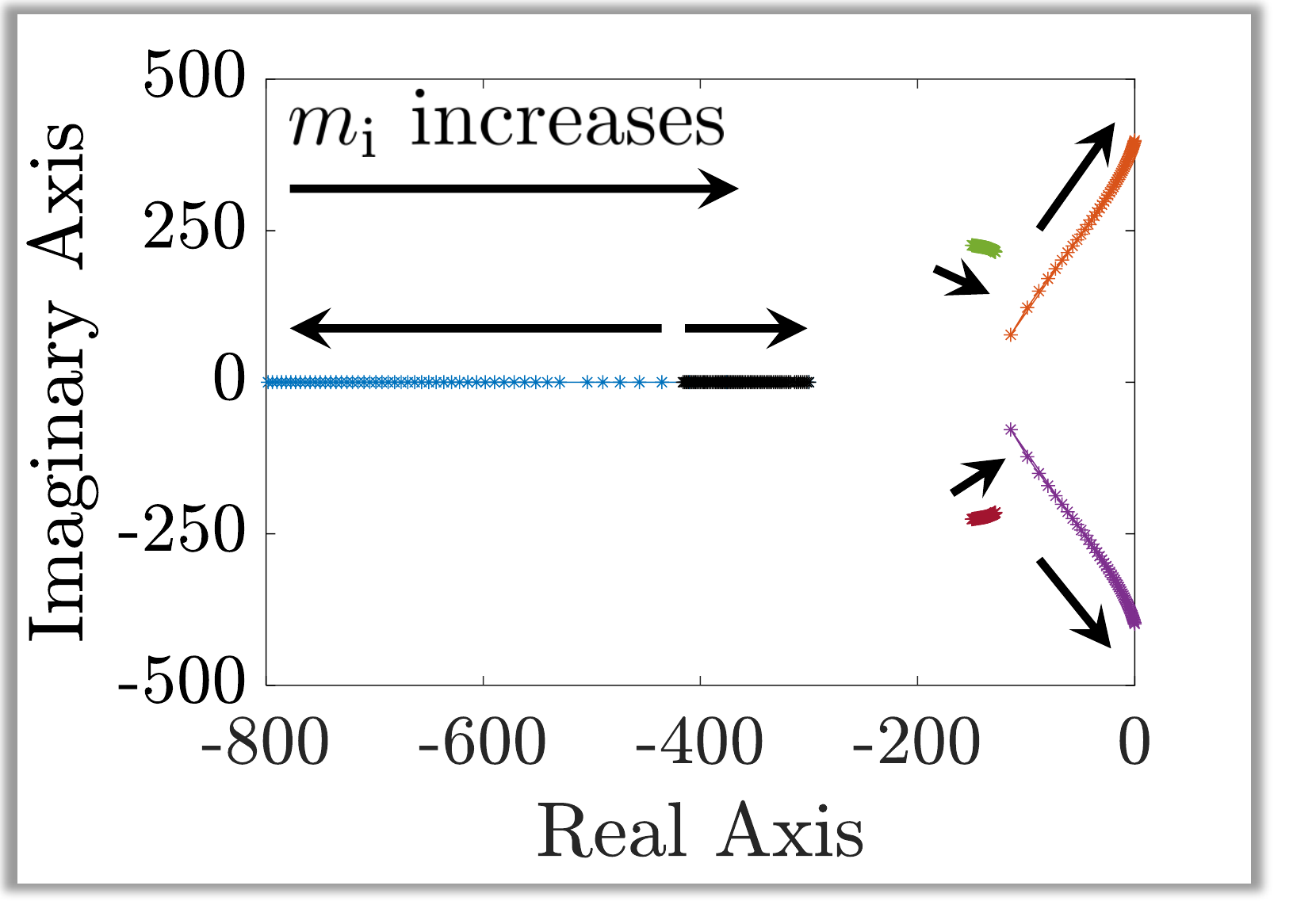}%
		\label{locusmisland}}
		
	\subfloat[]{\includegraphics[scale=0.25,trim={0cm 0cm 17cm 7.5cm},clip]{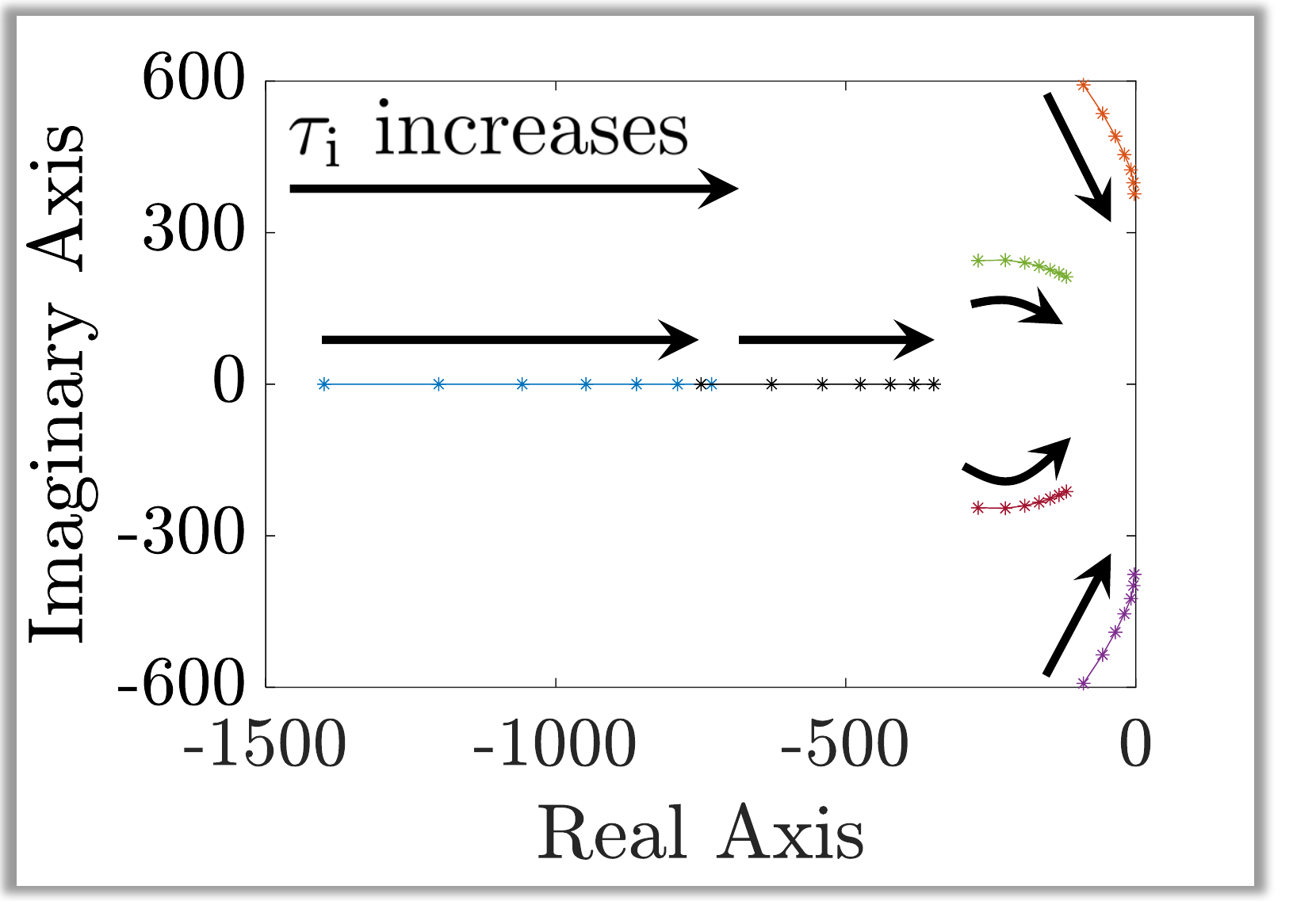}%
		\label{locustauisland}}~
	\subfloat[]{\includegraphics[scale=0.25,trim={0cm 0cm 17cm 7.5cm},clip]{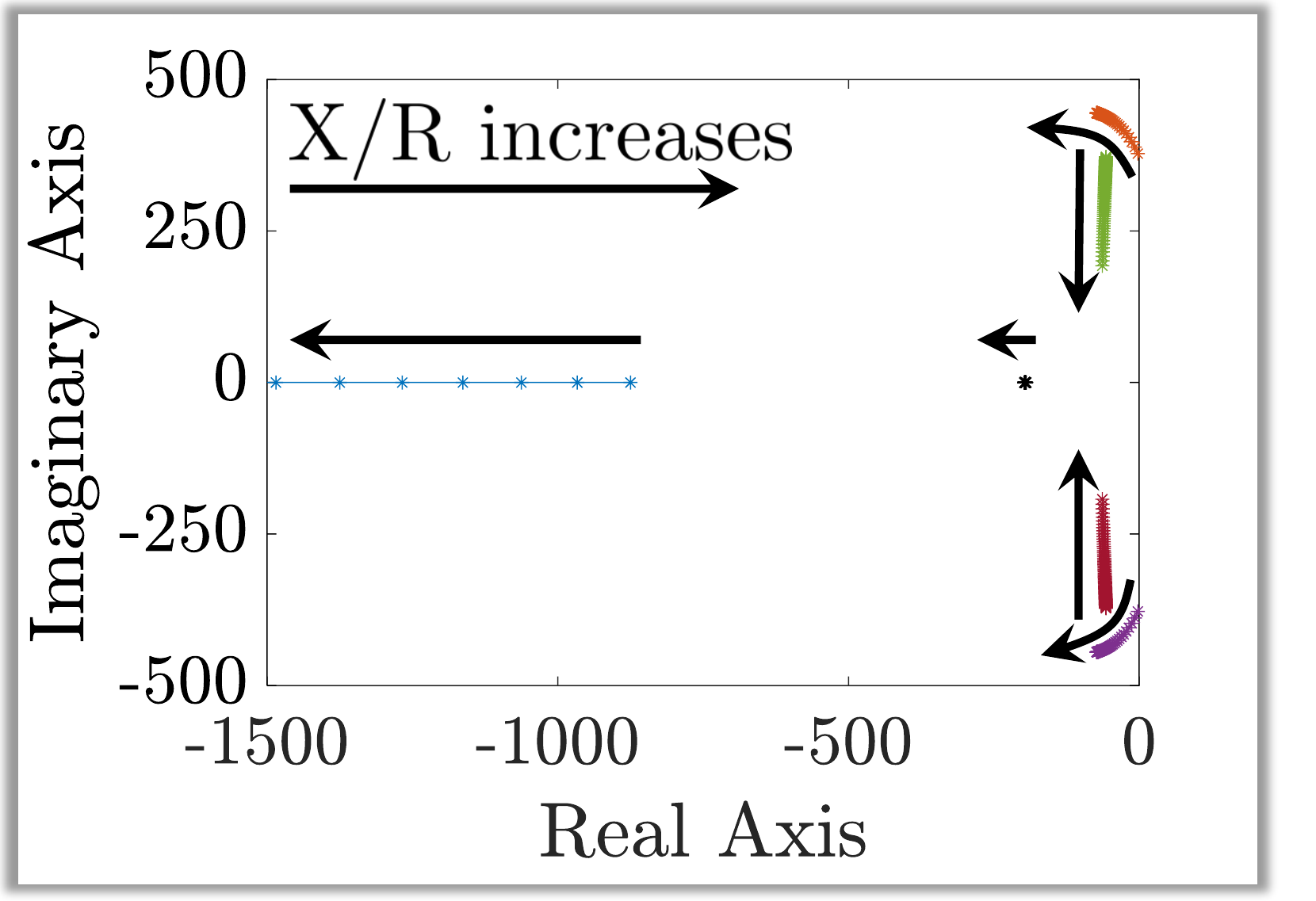}%
		\label{locusxbyrisland}}
	\caption{Eigenvalue spectrum of the system in \eqref{redlinearoffgrid} with, (a) $1.04\times 10^{-2}\le n_\mathrm{i} \le 3.64\times 10^{-2}$~rad/s/kW, (b) $41.67\times 10^{-3} \le m_\mathrm{i} \le 416.6\times 10^{-3}$~volt/kVAr, (c) $24.7 \le \tau_\mathrm{S,i} \le 41.2$~ms, (d) $0.33 \le X/R \le 3$.}
	\label{fig:locusisland}
\end{figure}
\subsection{System operating in off-grid mode}
While the system of Fig.~\ref{fig:system} is operating in off-grid mode, i.e. $\mathcal{I}_\mathrm{gs}=0$, it can be modelled as following non-linear system:
\begin{align}\label{offgridmodel}
    \dfrac{\mathrm{d}\underline{x}_\mathrm{off}}{\mathrm{d}t}=\underline{\mathcal{G}}_\mathrm{off}(\underline{x}_\mathrm{off}),
\end{align}
where, $\underline{\mathcal{G}}_\mathrm{off}(.)$ consists of \eqref{thetaID0}, \eqref{freqID0a}, \eqref{freqID0b} \eqref{voltID0a}, \eqref{voltID0b}, \eqref{ida}, \eqref{idb} for $\mathrm{i}$=$1,2$ and \eqref{load} with $\mathcal{I}_\mathrm{gs}=0$. Here, $\underline{x}_\mathrm{off}^\top=[\theta_\mathrm{1}~\theta_\mathrm{2}~\omega_\mathrm{r,1}~\omega_\mathrm{r,2}~V_\mathrm{r,1}~V_\mathrm{r,2}~i_\mathrm{o,1}^\mathrm{d}~i_\mathrm{o,2}^\mathrm{d}~i_\mathrm{o,1}^\mathrm{q}~i_\mathrm{o,2}^\mathrm{q}]$. As a result, the complete system in off-grid mode is a $10$-order nonlinear autonomous system with \textit{state} vector $\underline{x}_\mathrm{off}$. It is to be noted here that the system in off-grid mode is similar to islanded microgrid with two \textit{conventional}-droop controlled VSIs in order to share the common load at the PCC. The following two remarks are made before proceeding to the stability analysis of the system.
\begin{remark}\label{remark4}
In steady-state, \eqref{thetaID0} and \eqref{freqID0a} for $\mathrm{i}=1,2$ yield $\omega_\mathrm{r,1}=\omega_\mathrm{r,2}=\omega_\mathrm{PCC}$ which implies $n_\mathrm{1}P_\mathrm{1}=n_\mathrm{2}P_\mathrm{2}$ if and only if $n_\mathrm{1}P_\mathrm{ref,1}=n_\mathrm{2}P_\mathrm{ref,2}$. This is aligned with the conventional findings of droop controlled microgrids in literature \cite{chandorkar}.
\end{remark}
\begin{remark}\label{remark5}
Moreover, it is also common that the reactive power sharing accuracy with \textit{conventional}-droop control is strongly affected by line parameters \cite{gurrero1}. In light of this, if line parameters are neglected then \eqref{ida} and \eqref{idb} yield $\theta_\mathrm{1}=\theta_\mathrm{2}=0$ which implies $V_\mathrm{r,1}=V_\mathrm{r,2}$. Hence, from \eqref{voltID0a}, it can be stated that $m_\mathrm{1}Q_\mathrm{1}=m_\mathrm{2}Q_\mathrm{2}$ \textit{iff} $m_\mathrm{1}Q_\mathrm{ref,1}=m_\mathrm{2}Q_\mathrm{ref,2}$ \cite{chandorkar}.
\end{remark}
\begin{remark}
In conclusion with Remark~\ref{remark2}, \ref{remark3}, \ref{remark4}, and \ref{remark5}, it can be claimed that the proposed control for SEPSS, a seamless power flow can be achieved for CI. In normal on-grid mode, the grid alone supplies the total loads of CI. In case of emergency, the CI is transferred to off-grid mode seamlessly and rapidly and battery-fed VSIs are responsible to provide stable voltage and frequency with controlled power sharing. 
\end{remark}
Similarly, the problem for analyzing stability of the system in off-grid mode can be formulated as follows:
\begin{problem}\label{problem2}
Consider the system of \eqref{offgridmodel} linearized around an equilibrium point, $\underline{x}_\mathrm{off}^\mathrm{eq}$
\begin{align}\label{linearoffgrid}
    \dfrac{\mathrm{d}\Delta \underline{x}_\mathrm{off}}{\mathrm{d}t}=\underbrace{\mathcal{\underline{F}}_\mathrm{off}(\underline{x}_\mathrm{off})\big|_\mathrm{\underline{x}_\mathrm{off}^\mathrm{eq}}}_{\mathbf{A}_\mathrm{{off}}}\Delta \underline{x}_\mathrm{off},
\end{align}
where $\mathbf{A}_\mathrm{off}$ is the Jacobian matrix of $\mathcal{\underline{F}}_\mathrm{off}(.)$, the vector field of \eqref{offgridmodel}. Develop a systematic framework for tuning the control parameters i.e. $n_\mathrm{1}$, $n_\mathrm{2}$, $m_\mathrm{1}$, $m_\mathrm{2}$, $\tau_\mathrm{1}$, $\tau_\mathrm{2}$, so that asymptotic stability of the equilibrium of \eqref{linearoffgrid} is guaranteed.
\end{problem}
In similar way by suppressing the EM dynamics of the lines and loads with the assumption of reaching the fast states i.e. $\underline{x}_\mathrm{off}^\mathrm{f}:=[i_\mathrm{o,1}^\mathrm{d}~i_\mathrm{o,2}^\mathrm{d}~i_\mathrm{o,1}^\mathrm{q}~i_\mathrm{o,2}^\mathrm{q}]^\top$ to a  quasi-steady-state values (i.e. $\underline{x}_\mathrm{off,ss}^\mathrm{f}$), the following reduced-order model of the system \eqref{linearoffgrid} with slow states, $\underline{x}_\mathrm{off}^\mathrm{s}:=[\theta_\mathrm{1}~\theta_\mathrm{2}~\omega_\mathrm{r,1}~\omega_\mathrm{r,2}~V_\mathrm{r,1}~V_\mathrm{r,2}]^\top$ can be formulated (more elaborated discussion in Appendix~B):
\begin{align}\label{redlinearoffgrid}
        \dfrac{\mathrm{d}\Delta \underline{x}_\mathrm{off}^\mathrm{s}}{\mathrm{d}t}=\mathbf{A}_\mathrm{\textbf{off}}^\mathrm{\textbf{s}}\Delta \underline{x}_\mathrm{off}^\mathrm{s}.
\end{align}
\renewcommand{\arraystretch}{1.2}
\begin{table}[t]
\centering
\caption{DOMAIN OF CONTROL PARAMETERS FOR $\mathrm{i}=1,2$}
\label{table:control}
\begin{tabular}{|c|c|c|}
\hline 
\textbf{Parameters} & \textbf{Range} & \textbf{Selected Value}    \\ \hline \hline
$n_\mathrm{i}$ (rad/s/kW) & $[1.04~~3.64]\times 10^{-2}$ & $2.08\times 10^{-2}$  \\ \hline
$m_\mathrm{i}$ (volt/kVAr) & $[41.67~~416.6]\times 10^{-3}$ & $208.3\times 10^{-3}$ \\ \hline
$m_\mathrm{int,i}$ (volt/s/kVAr) & $[0.6~~0.84]$ & $0.67$ \\ \hline
$\tau_\mathrm{i}$ (ms) & $[24.7~~41.2]$ & $33$\\ \hline
\end{tabular}
\end{table}
With the equilibrium point of the model in \eqref{redlinearoffgrid} and the parameters indicated in Table~\ref{table:control} and \ref{table:data}, the root locus-based analysis is conducted in order to obtain a range of various control parameters, as stated in Problem~\ref{problem2} empirically. With the assumption of having homogeneous control parameters, i.e. $n_\mathrm{1}=n_\mathrm{2}$, $m_\mathrm{1}=m_\mathrm{2}$ and $\tau_\mathrm{1}=\tau_\mathrm{2}$, the complete eigenvalue spectrum of the system state matrix $\mathbf{A}^\mathrm{\textbf{s}}_\mathrm{\textbf{off}}$ are shown in Fig.~\ref{fig:locusisland}\subref{locusnisland}, \ref{fig:locusisland}\subref{locusmisland}, \ref{fig:locusisland}\subref{locustauisland} and \ref{fig:locusisland}\subref{locusxbyrisland}. It is important that the sensitivity of droop control in island mode w.r.t. the line types (i.e. $X/R$ ratio of lines) is also added in this stability study. Similar to the system with on-grid mode, it is observed here also that with the increase of these control parameters individually, some of the low-frequency eigen values gradually move away from the real axis, meanwhile close to the imaginary axis until it reach to unstable right-half zone. Similarly, some of the stable eigenvalues gradually move on the real axis only away from the imaginary axis. It is observed that the same range of common control parameters (i.e. $n_\mathrm{i}$'s, $m_\mathrm{i}$'s and $\tau_\mathrm{i}$'s for $\mathrm{i}=1, 2$), affecting the microgrid stability and dynamic performance are obtained from the critical value of the root locus as tabulated in Table~\ref{table:control}.
\subsection{System during on-grid to off-grid transition}
Apart from the steady-state operations in on- and off-grid mode, the system of Fig.~\ref{fig:system} will occasionally experience the transient operations during on- to off-grid transition and vice versa. Island detection algorithms are important aspects to be considered for on- to off-grid transition. On the other hand, grid re-synchronization algorithm is necessary for off- to on-grid transition. Although the response during off- to on-grid transition will dominantly depends on the grid behavior which is not the primary focus of this section, the analysis of the behavior of the system during on- to off-grid transition is an important task while guaranteeing seamless transition. This study will showcase the benefit of the proposed control by having seamless and rapid power flow recovery of the system under study after grid failure. This study is achieved by identifying the influencing factor of control parameters (i.e. $n_\mathrm{i}$, $m_\mathrm{i}$, $\tau_\mathrm{S,i}$) on transient response. While the system of Fig.~\ref{fig:system} is operating in on-grid mode ($\mathcal{I}_\mathrm{gs}=1$) before the transition, equivalently the system given by \eqref{ongridmodel} is operating at steady-state given by $\underline{x}_\mathrm{on}=\underline{x}_\mathrm{on}^\mathrm{eq}$. Therefore, during the transition, the system given by \eqref{offgridmodel} will have a trajectory staring from initial conditions of the states, $\underline{x}_\mathrm{off}=\underline{x}_\mathrm{off}^\mathrm{initial} \subset \underline{x}_\mathrm{on}^\mathrm{eq}$, (where $\subset$ denotes the common states between $\underline{x}_\mathrm{off}$ and $\underline{x}_\mathrm{on}$) to the steady-state values of the system in off-grid mode, given by \eqref{offgridmodel}, i.e. $\underline{x}_\mathrm{off}=\underline{x}_\mathrm{off}^\mathrm{eq}$. As evidenced from \eqref{ongridmodel} and \eqref{offgridmodel}, the system is highly non-linear in nature and as a result $\underline{x}_\mathrm{on}^\mathrm{eq}$ is determined numerically by solving $\mathcal{\underline{G}}_\mathrm{on}(\underline{x}_\mathrm{on})=0$ using Newton–Raphson method and Forward Euler method is used for solving \eqref{offgridmodel} with given initial conditions. The main objective is to identify the control parameter that has significant impact on the response time of the system during transition and it is achieved by observing the response of $\mathrm{i}^{th}$ VSI output active power, $p_\mathrm{i}=f_\mathrm{P,i}(\underline{x}_\mathrm{off})$, and reactive power, $q_\mathrm{i}=f_\mathrm{Q,i}(\underline{x}_\mathrm{off})$, given by \eqref{freqID0b} and \eqref{voltID0b} respectively. It is observed that among the control parameters, the time constant, $\tau_\mathrm{S,i}$, of the low-pass filter in the power controller has dominating effect on the transient response time during the transition as showcased in Fig.~\ref{fig:resultjump1} and Fig.~\ref{fig:resultjump2}. Clearly before the transition ($t=0.05~s$), the system given by \eqref{ongridmodel} is operating in steady-state and as a result, both $p_\mathrm{i}$ and $q_\mathrm{i}$, determined by $f_\mathrm{P,i}(\underline{x}_\mathrm{on}^\mathrm{eq})$, $f_\mathrm{Q,i}(\underline{x}_\mathrm{on}^\mathrm{eq})$ respectively, (given that no load change occurs during this transition), are zero. This result is in alignment with the Theorem~\ref{theorem1}. From $t=0.05~s$ onward, the system given by \eqref{offgridmodel} starts from the initial conditions and eventually reaches to new steady-state, $\underline{x}_\mathrm{off}^\mathrm{eq}$, after going through a transient. As a result, both $p_\mathrm{i}$ and $q_\mathrm{i}$ become constant determined by $f_\mathrm{P,i}(\underline{x}_\mathrm{off}^\mathrm{eq})$ and  $f_\mathrm{Q,i}(\underline{x}_\mathrm{off}^\mathrm{eq})$ respectively. It is observed here that the increase in time constant of power controller filter (same as Table~\ref{table:control}) results in an increase in response time of power recovery. 
\begin{figure}[t]
	\centering
    \includegraphics[scale=0.26,trim={0cm 0cm 0cm 6cm},clip]{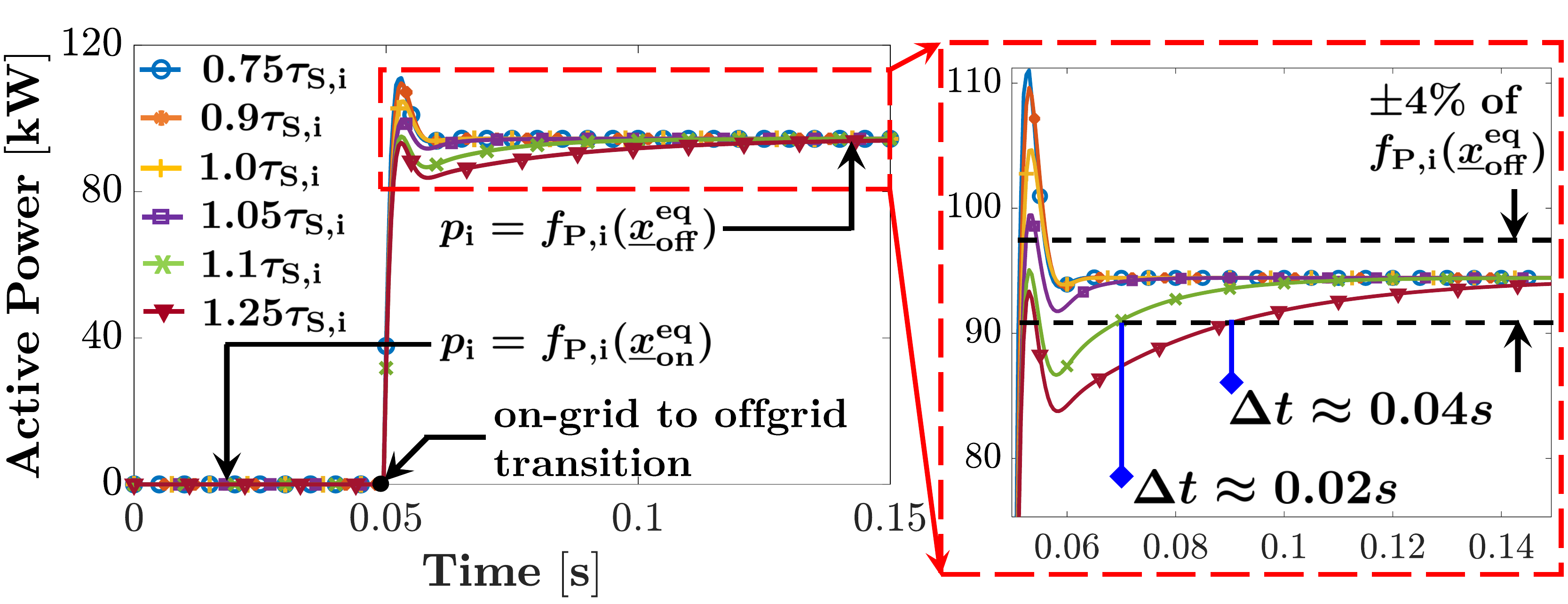}%
	\caption{Active power output of $\mathrm{i}^{th}$ VSI during on-grid to off-grid transition.}
	\label{fig:resultjump1}
\end{figure}
\begin{figure}[t]
	\centering
    \includegraphics[scale=0.26,trim={0cm 0cm 0cm 6cm},clip]{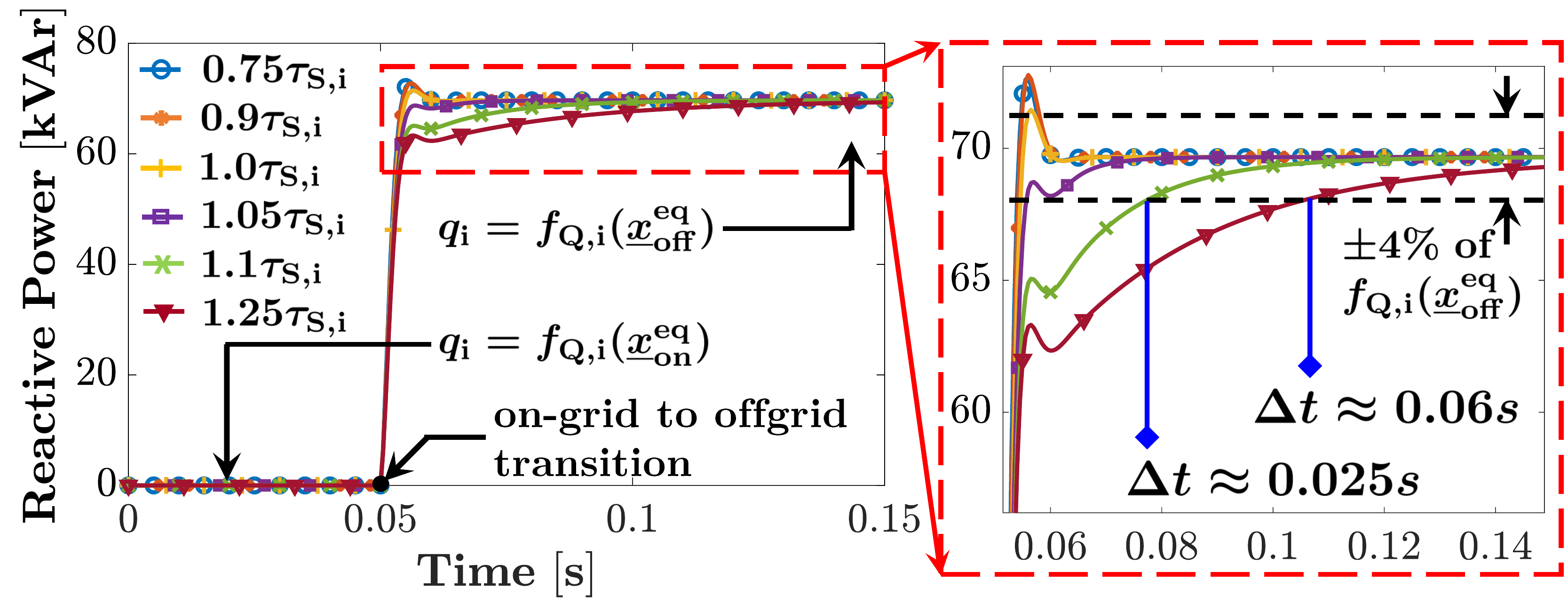}%
	\caption{Reactive power output of $\mathrm{i}^{th}$ VSI during on-grid to off-grid transition.}
	\label{fig:resultjump2}
\end{figure}
\section{Experimental Results and Verification}\label{result}
\begin{figure}[t]
	\centering
    \includegraphics[scale=0.26,trim={0cm 0cm 0cm 0cm},clip]{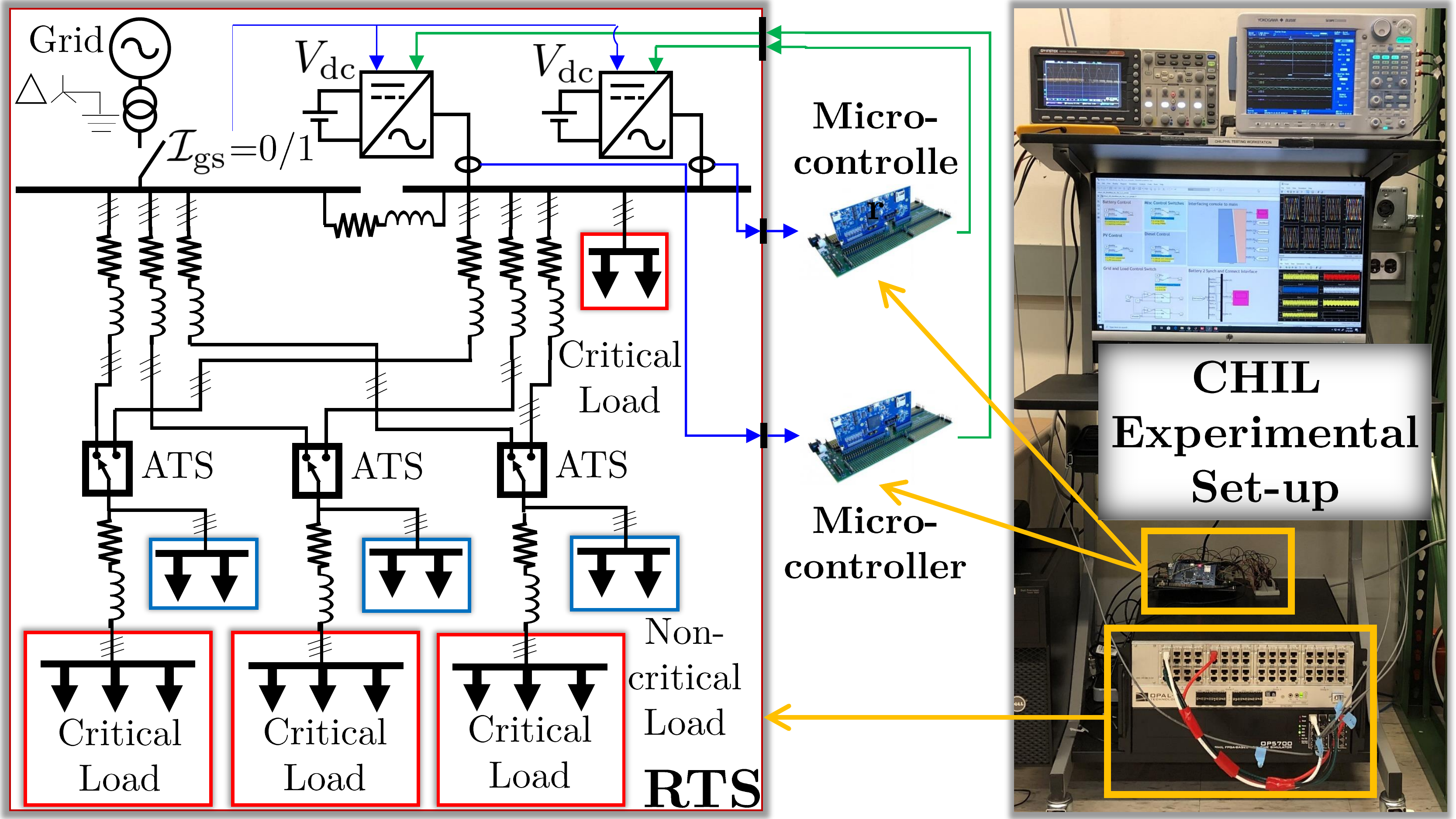}%
	\caption{RTS based controller-hardware-in-the-loop simulation platform.}
	\label{fig:setup}
\end{figure}
\renewcommand{\arraystretch}{1.2}
\begin{table}[t]
\centering
\caption{$3$-PHASE VSI AND NETWORK PARAMETERS UNDER STUDY}
\label{table:data}
\begin{tabular}{|c|c|}
\hline 
\textbf{VSI} & \textbf{Value}    \\ \hline \hline
Ratings ($3$-$\phi$) & $480$~V (L-L), $60$~Hz, $120$~kVA, $0.85$~pf \\ \hline
VSI Parameters & $V_\mathrm{dc}$ = $1000$~V, $f_\mathrm{sw}$ = $10$~kHz \\ \hline
Filter Parameters & $L_\mathrm{f}$ = $150$~$\mu$H, $L_\mathrm{g}$ = $15$~$\mu$H, $C_\mathrm{f}$ = $110~\mu$F \\ \hline
Line Parameter & $R_\mathrm{Line,i}=0.55$~m$\Omega$, $L_\mathrm{Line,i}=0.2$~mH,~$\mathrm{i}=1,2$  \\ \hline \hline \hline
\textbf{Network} & \textbf{Value}    \\ \hline \hline
Grid Rating ($3$-$\phi$) &  $13.8$~kV fed $\mathrm{Y}$-$\mathrm{Y}$ $0.75$~MVA, $13.8/0.48$~kV Xr. \\ \hline
Line Parameter & $R_\mathrm{l,g}=5$~m$\Omega$, $L_\mathrm{l,g}=30$~$\mu$H \\ \hline
\end{tabular}
\end{table}
\begin{figure*}[t]
	\centering
    \includegraphics[scale=0.55,trim={0cm 0cm 3cm 12cm},clip]{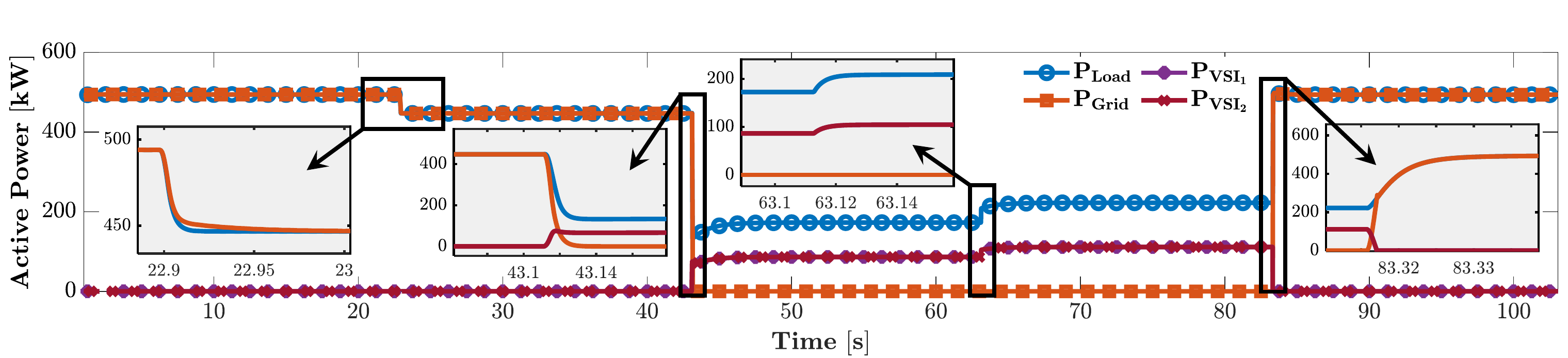}%
	\caption{Profile of active power flow in critical infrastructure of Fig.~\ref{fig:setup} during the sequence of events considered in test case.}
	\label{fig:P}
\end{figure*}
\begin{figure*}[t]
	\centering
    \includegraphics[scale=0.55,trim={0cm 0cm 3cm 12cm},clip]{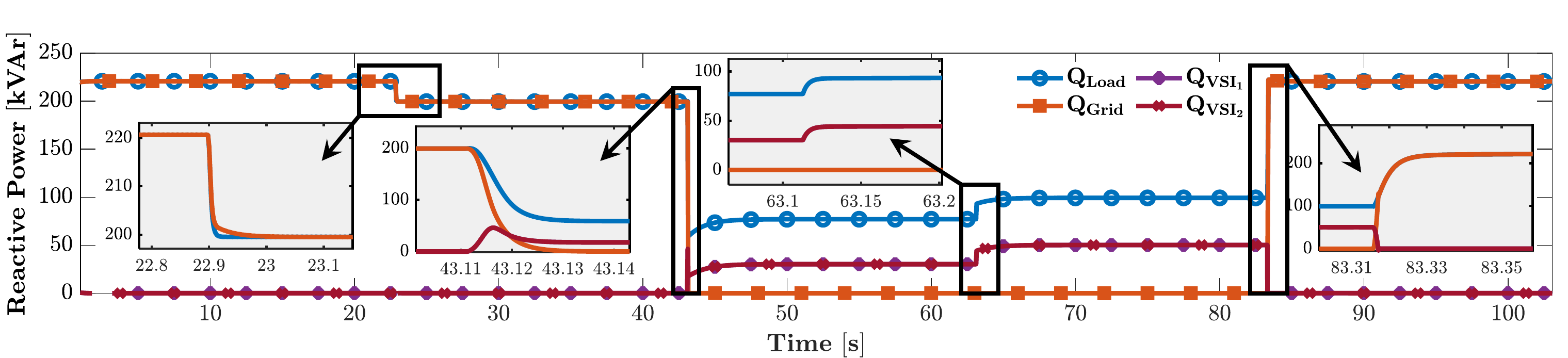}%
	\caption{Profile of reactive power flow in critical infrastructure of Fig.~\ref{fig:setup} during the sequence of events considered in test case.}
	\label{fig:Q}
\end{figure*}
\begin{figure*}[t]
	\centering
    \includegraphics[scale=0.55,trim={0cm 0cm 3cm 12cm},clip]{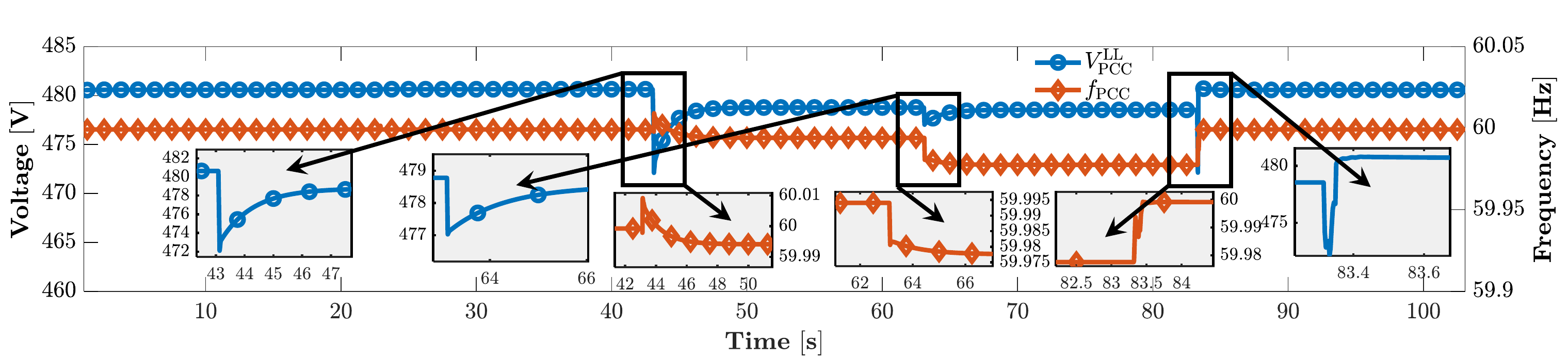}%
	\caption{Profile of voltage RMS (L-L) and frequency of critical infrastructure of Fig.~\ref{fig:setup} during the sequence of events considered in test case.}
	\label{fig:VF}
\end{figure*}
In order to substantiate the proposed strategy of seamless power recovery for CIs, controller hardware-in-the-loop (CHIL)-based real-time (RT) simulation studies are conducted as shown in Fig.~\ref{fig:setup}. The test-system of CI under study is developed based on the electrical network of M-Health Fairview, University of Minnesota Medical Center with certain relaxations and modifications \cite{umn}. It has a peak load of $500$kW distributed as non-critical ($300$kW) and critical loads ($200$kW). The parameters are tabulated in Table~\ref{table:data}. The network along with power circuits of two VSIs (parameters tabulated in Table~\ref{table:data}) is emulated using eMEGASIM platform inside the OP$5700$ RT-simulator (RTS) manufactured by OPAL-RT. The control loops of both the VSIs are realized on two low-cost Texas-Instruments TMS$28379$D, $16$/$12$-bit floating-point $200$~MHz Delfino micro-controller boards. The control parameters used for validation are tabulated in Table~\ref{table:control}. 
\par In order to validate the proposed control architecture, a test case is examined by conducting a sequence of events in the RT platform. The CI is initially operating in on-grid mode supplying both critical and non-critical loads with both the VSIs synchronized and connected to the network. At $t\approx22.91$s, there is a load drop in CI and subsequently at $t\approx43.18$s, CI is disconnected from the grid and is operating in off-grid mode until $t\approx83.32$s when the grid returns. In off-grid mode, both VSIs are supplying the critical loads with a load jump at $t\approx63.11$s. Moreover, once the grid comes back \cite{gridsync}, both non-critical and critical loads of CI are supplied by the grid again. 
\par Fig.~\ref{fig:P} and \ref{fig:Q} illustrate the active and reactive power flow in the considered test case. Fig.~\ref{fig:iB1}, \ref{fig:iB2} and \ref{fig:iG} illustrate the three-phase instantaneous current waveform supplied by the VSI $1$ and $2$ and the grid respectively. It is observed here that even if both the VSIs are connected to the network, during on-grid mode, the output active and reactive power of VSIs (i.e. $\mathrm{P}_\mathrm{VSI_\mathrm{1}}$, $\mathrm{P}_\mathrm{VSI_\mathrm{2}}$, $\mathrm{Q}_\mathrm{VSI_\mathrm{1}}$ and $\mathrm{Q}_\mathrm{VSI_\mathrm{2}}$) are zero until $t\approx43.18$s when CI is disconnected from grid and operating in off-grid mode. The complete active and reactive power demand of the CI (i.e. $\mathrm{P}_\mathrm{Load}$, $\mathrm{Q}_\mathrm{Load}$) is supplied by the grid ($\mathrm{P}_\mathrm{Grid}$, $\mathrm{Q}_\mathrm{Grid}$) alone. Therefore it can be stated that the proposed droop controller is making sure not to utilize VSIs while grid is available in order to supply both non-critical and critical load of the amount of $500$kW and $220$kVAr. Once the CI is disconnected from grid and operating in off-grid mode, only critical loads are served via both VSIs which is $250$kW and $100$kVAr at peak. It is observed that both the VSIs are equally sharing the total active and reactive power demand of the critical loads of the CI as their droop characteristics are designed to share the load demand symmetrically. Moreover, at $t\approx83.32$s when grid returns, the total critical and non-critical loads of the CI is again served by the grid alone and output active and reactive power of the VSIs become zero. 
\par During each transient event, the transient performance of both active and reactive power curves are also provided in Fig.~\ref{fig:P} and \ref{fig:Q}. The selection of droop coefficients in the experiment makes sure to provide large damping and less oscillatory behavior in power transients which results a significantly fast transient behavior in order to reach steady state ($\le 0.1$s) as studied in the eigen analysis of the system. Two events are required to be analyzed separately here and these are the on-grid/off-grid and off-grid/on-grid transitions. During on-grid to off-grid transition, the seamless power flow is achieved as both the VSIs are always connected to the system and requirement of synchronisation issues are eliminated. There is no time of electrical separation of the CI network and the power flow recovery is reaching to steady-state within $0.2$s as shown in Fig.~\ref{fig:P} and \ref{fig:Q}. In similar way, when the grid returns the control architecture will forcefully make output active and reactive power of both VSIs to drop down to zero and this is achieved within $0.2$s as evidenced in Fig.~\ref{fig:P} and \ref{fig:Q}. 
\par Fig.~\ref{fig:VF} illustrates the voltage RMS and frequency at PCC of the CI during these sequence of events with magnified plot during transients. It is observed that the voltage RMS during the events are within $\pm 10$V which is around $2\%$ of nominal value. Moreover, the frequency of the network also stays within $\pm 0.03$Hz. The instantaneous three-phase voltage waveform (L-N) at the PCC of the considered CI during the sequence of events is illustrated in Fig.~\ref{fig:vNET}. Along with the steady-state waveform, the profile of instantaneous voltage waveform during each transients are also smooth and remain under the CBEMA/ITIC curve as reported in IEEE~$446$ \cite{ieee446}. It is obvious that while the CI is in on-grid mode, the frequency will be dictated by the stiff grid. However, during off-grid mode, the VSIs are responsible to hold a stable frequency as well as voltage in the network. It is observed that the voltage and frequency of the network of CI are always under the permissible range of voltage and frequency regulation limits recommended in accredited standard IEEE~$446$ \cite{ieee446}. 
\begin{figure}[t]
	\centering
    \includegraphics[scale=0.27,trim={0cm 0cm 3cm 7cm},clip]{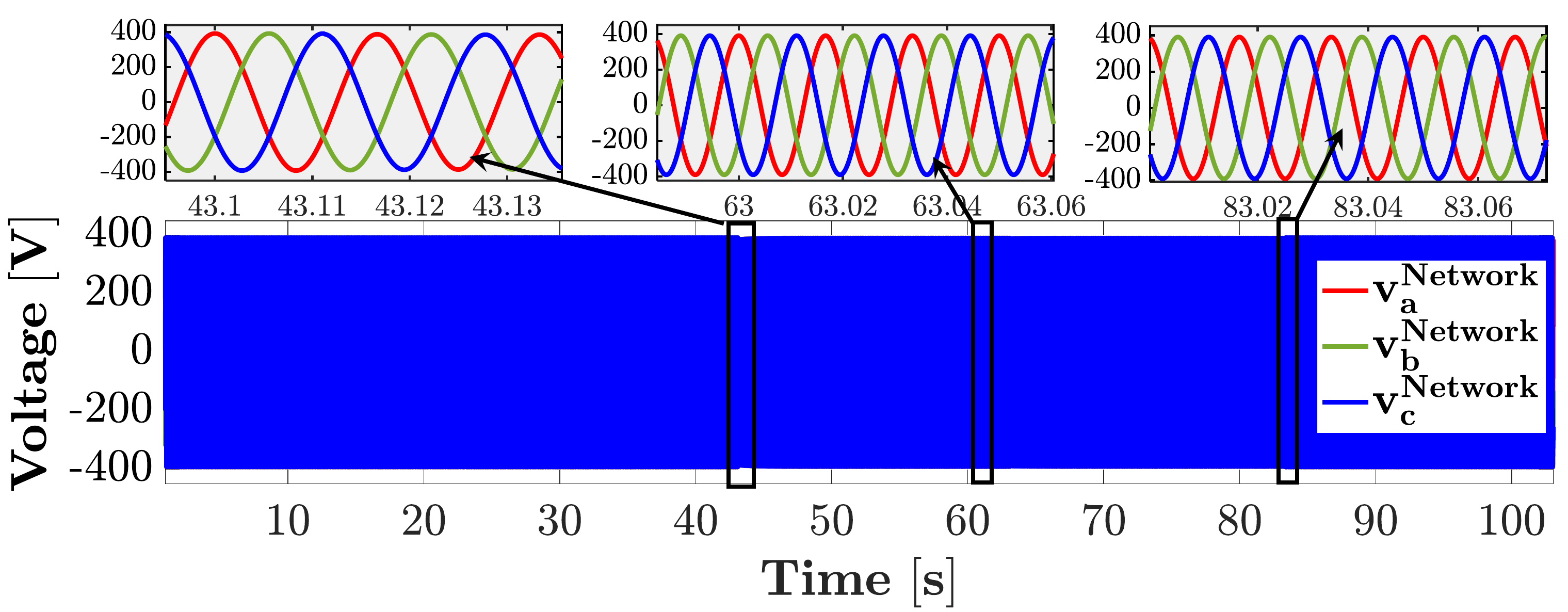}%
	\caption{Instantaneous three-phase voltage waveform (L-N) measured at the PCC of the CI during the sequence of event.}
	\label{fig:vNET}
\end{figure}
\begin{figure}[t]
	\centering
    \includegraphics[scale=0.27,trim={0cm 0cm 3cm 7cm},clip]{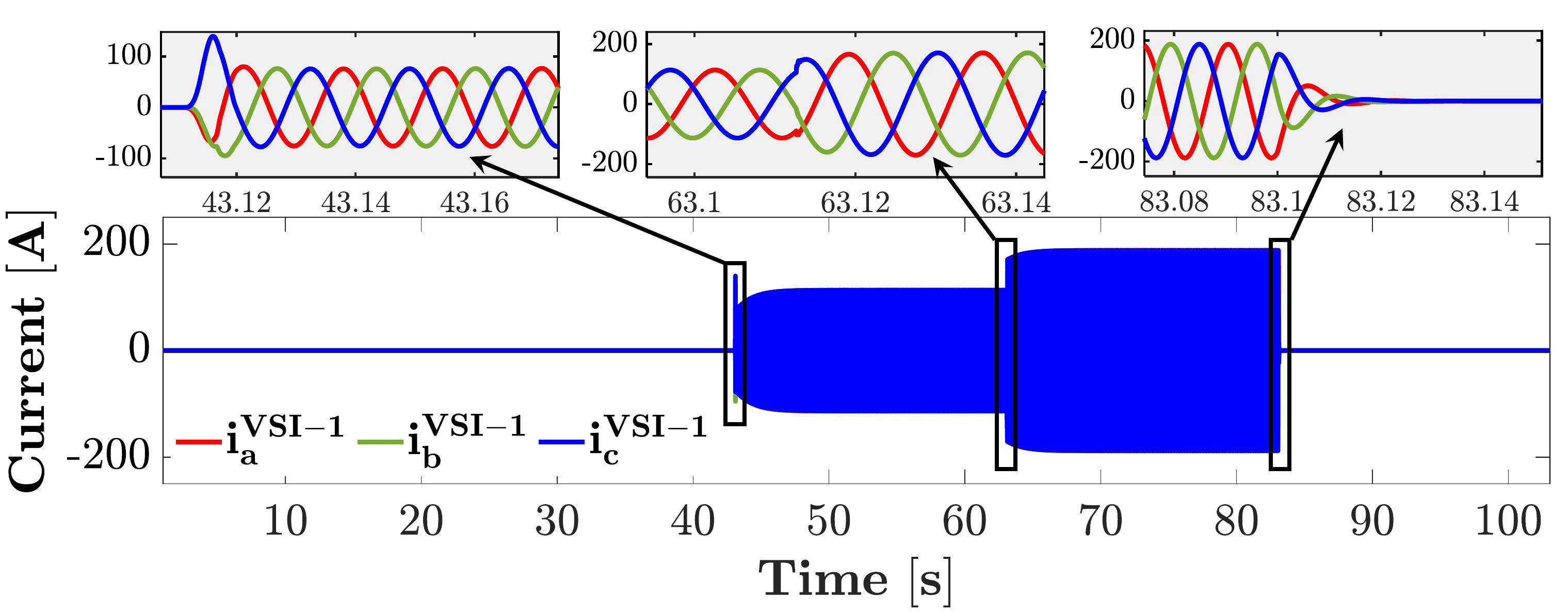}%
	\caption{Instantaneous three-phase current waveform supplied by VSI $1$ of the CI during the sequence of event.}
	\label{fig:iB1}
\end{figure}
\begin{figure}[t]
	\centering
    \includegraphics[scale=0.27,trim={0cm 0cm 3cm 7cm},clip]{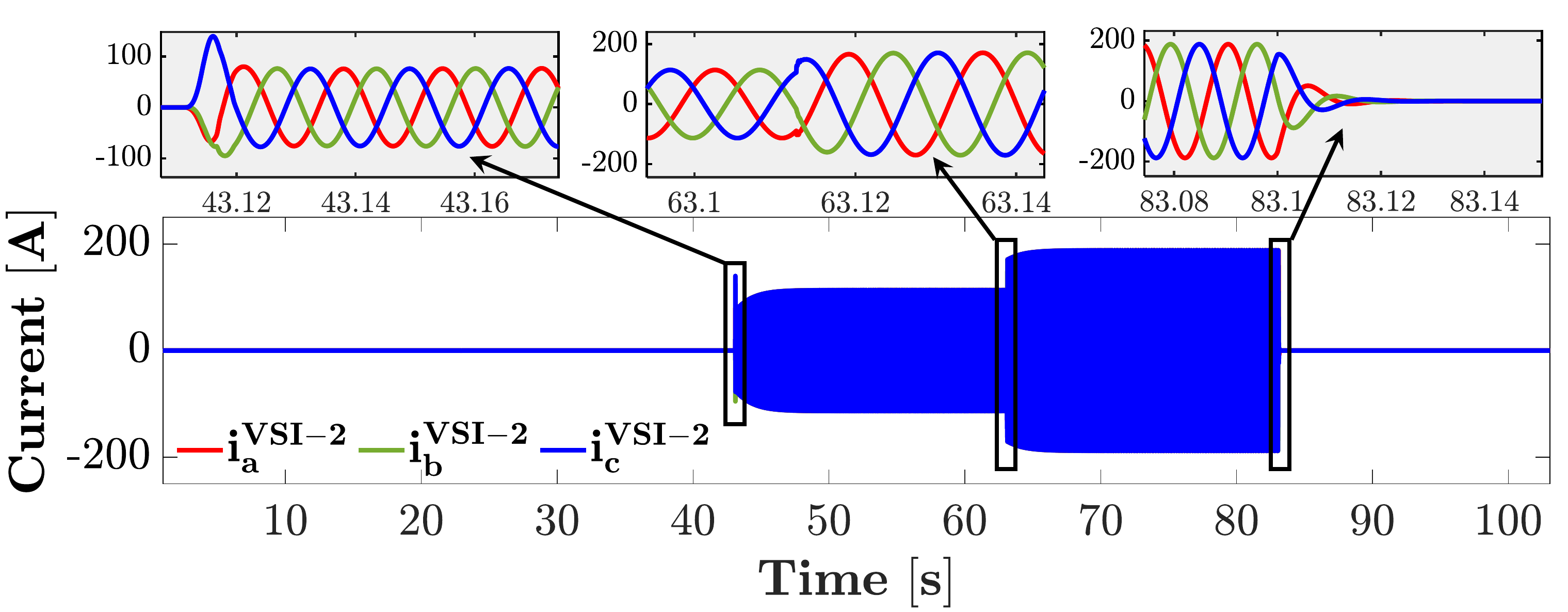}%
	\caption{Instantaneous three-phase current waveform supplied by VSI $2$ of the CI during the sequence of event.}
	\label{fig:iB2}
\end{figure}
\begin{figure}[t]
	\centering
    \includegraphics[scale=0.27,trim={0cm 0cm 3cm 7cm},clip]{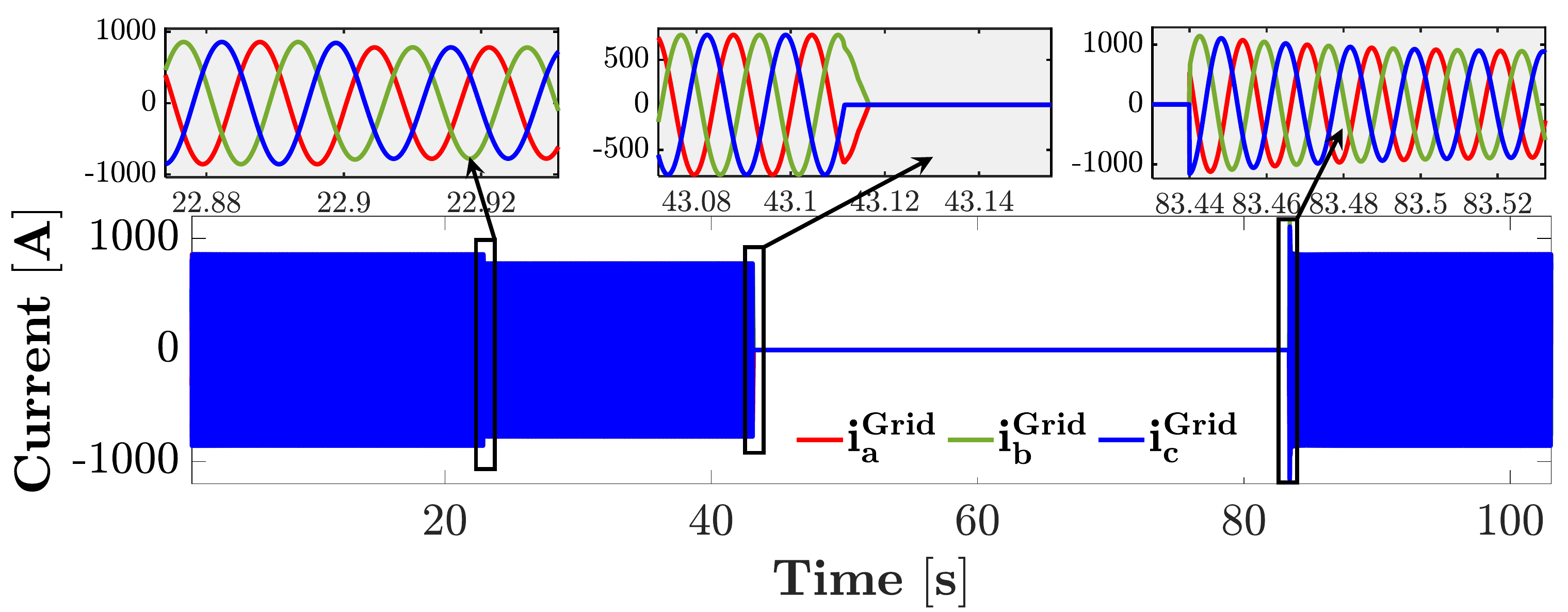}%
	\caption{Instantaneous three-phase current waveform supplied by grid of the CI during the sequence of event.}
	\label{fig:iG}
\end{figure}
\section{Conclusion}\label{conclusion}
In this paper, a rapid and seamless power recovery strategy for SEPSS of CI by a mode-dependent droop control-based VSI is proposed. With this proposed droop control it is shown, both analytically and experimentally, that none of the VSIs supplied power to the network while operating in on-grid mode and entire load of CI is supplied alone by the grid as desired. Whereas, during off-grid mode both the VSIs share the required demand among themselves with a seamless transition from on-grid mode following a grid failure. These two accomplishments are in alignment with the synopsis of the SEPSS for any CI. Moreover, stability of such control in system under study is guaranteed by systematic modeling of the system. Viability of the proposed control is corroborated by controller hardware-in-the-loop-based real-time simulation studies on a realistic CI electrical network based on a healthcare facility.
\section*{APPENDIX}
\subsection*{Appendix~A: [Proof of Theorem \ref{theorem1}]}
\begin{proof}
While the system, shown in Fig.~\ref{fig:system}, is operating in on-grid mode, i.e. $\mathcal{I}_\mathrm{gs}$=$1$ and grid is assumed to be at nominal (inferred from Assumption~\ref{assumption3}), the complete system model can be formulated by \eqref{ongridmodel}. Therefore, in steady-state, 
\begin{align*}
    \underline{\mathcal{G}}_\mathrm{on}(\underline{x}_\mathrm{on})\Big|_{\underline{x}_\mathrm{on}^\mathrm{eq}}=0.
\end{align*}
Hence, \eqref{thetaID0} with $i$=$1,2$ and \eqref{thetagrid} yield $\omega_\mathrm{r,1}=\omega_\mathrm{r,2}=\omega_\mathrm{grid} = \omega_\mathrm{nom}$. However, from \eqref{freqID1a} for $i$=$1,2$, $\omega_\mathrm{r,1}=\omega_\mathrm{nom}-n_\mathrm{1}P_\mathrm{1}$ and $\omega_\mathrm{r,2}=\omega_\mathrm{nom}-n_\mathrm{2}P_\mathrm{2}$. Hence, this concludes that $n_\mathrm{1}P_\mathrm{1}=n_\mathrm{2}P_\mathrm{2}=0$ which implies $P_\mathrm{1}=P_\mathrm{2}=0$.\newline
Similarly, \eqref{psi} for $i$=$1,2$ yields $V_\mathrm{r,1}=V_\mathrm{nom}-m_\mathrm{int,1}\psi_\mathrm{1}^\mathrm{Q}$ and $V_\mathrm{r,2}=V_\mathrm{nom}-m_\mathrm{int,2}\psi_\mathrm{2}^\mathrm{Q}$. Moreover, from \eqref{voltID1a} for $i$=$1,2$ it follows that $m_\mathrm{1}Q_\mathrm{1}=k_\mathrm{m,1}(V_\mathrm{nom}-V_\mathrm{r,1}-m_\mathrm{int,1}\psi_\mathrm{1}^\mathrm{Q})$ and $m_\mathrm{2}Q_\mathrm{2}=k_\mathrm{m,2}(V_\mathrm{nom}-V_\mathrm{r,2}-m_\mathrm{int,2}\psi_\mathrm{2}^\mathrm{Q})$. As a result, $m_\mathrm{1}Q_\mathrm{1}=m_\mathrm{2}Q_\mathrm{2}=0$ which implies $Q_\mathrm{1}=Q_\mathrm{2}=0$. Hence, active and reactive power of both the VSIs are zero and grid alone will supply the total load requirements of the CI.
\end{proof}
\subsection*{Appendix~B: [Time-scale Separation in \eqref{linearongrid} and \eqref{linearoffgrid}]}
The linearized system equation, given in \eqref{linearongrid} or \eqref{linearoffgrid}, can be written in \textit{standard singularly perturbed form} as follows:
\begin{align}\label{genlin}
 \begin{bmatrix}
  \dot{\underline{x}} \\ \epsilon\dot{\underline{z}} \end{bmatrix}
  =
  \begin{bmatrix}\mathbf{A}_\mathrm{xx} & \mathbf{A}_\mathrm{xz} \\ \mathbf{A}_\mathrm{zx} & \mathbf{A}_\mathrm{zz}
 \end{bmatrix}
  \begin{bmatrix}
  {\underline{x}} \\ {\underline{z}} \end{bmatrix},
\end{align}
where, $\underline{x}$ denotes the slow dynamic variables and $\underline{z}$ denotes the fast dynamic variables. For system of \eqref{linearongrid}, $\underline{x}=\Delta \underline{x}_\mathrm{on}^\mathrm{s}$ and $\underline{z}=\Delta \underline{x}_\mathrm{on}^\mathrm{f}$ and as a result, $\mathbf{A}_\mathrm{xx}\in \mathbb{R}^{9\times 9}$, $\mathbf{A}_\mathrm{xz}\in \mathbb{R}^{9\times 6}$, $\mathbf{A}_\mathrm{zx}\in \mathbb{R}^{6\times 9}$ and $\mathbf{A}_\mathrm{zz}\in \mathbb{R}^{6\times 6}$. Similarly, for system of \eqref{linearoffgrid}, $\underline{x}=\Delta \underline{x}_\mathrm{off}^\mathrm{s}$ and $\underline{z}=\Delta \underline{x}_\mathrm{off}^\mathrm{f}$ and as a result, $\mathbf{A}_\mathrm{xx}\in \mathbb{R}^{6\times 6}$, $\mathbf{A}_\mathrm{xz}\in \mathbb{R}^{6\times 4}$, $\mathbf{A}_\mathrm{zx}\in \mathbb{R}^{4\times 6}$ and $\mathbf{A}_\mathrm{zz}\in \mathbb{R}^{4\times 4}$. $\epsilon$'s are mostly dictated by the values of $L/R$ ratio of the lines which is small for the usual range of $X/R$ of distribution lines \cite{kundur}. Hence the perturbation parameters, $\epsilon$, are very small ($0 < \epsilon \ll 1$). It should be now obvious that the dynamic behavior of system is dictated by two timescales $t$, $(t/\epsilon)$ with the small parameter $\epsilon$ giving rise to these timescales. With non-singular $\mathbf{A}_\mathrm{zz}$ for both the systems, as $\epsilon \to 0$, the quasi-steady-state (QSS) solution of states, $\underline{z}$, obtained from \eqref{genlin} is
\begin{align}
    \underline{z}^\mathrm{ss} = -\mathbf{A}_\mathrm{zz}^{-1}\mathbf{A}_\mathrm{zx}\underline{x}= -\mathbf{A}^{'}_\mathrm{zx}\underline{x}.
\end{align}
Substitution of $\underline{z}$ in \eqref{genlin} by the QSS value, provides the reduced-order model as follows:
\begin{align}
    \dot{\underline{x}} = [\mathbf{A}_\mathrm{xx}-\mathbf{A}_\mathrm{xz}\mathbf{A}^{'}_\mathrm{zx}]\underline{x} = \mathbf{A}^{'}_\mathrm{xx}\underline{x}.
\end{align}
For system of \eqref{linearongrid}, $\mathbf{A}^{'}_\mathrm{xx}=\mathbf{A}^\mathrm{s}_\mathrm{on}$ of \eqref{redlinearongrid} and for system of \eqref{linearoffgrid}, $\mathbf{A}^{'}_\mathrm{xx}=\mathbf{A}^\mathrm{s}_\mathrm{off}$ of \eqref{redlinearoffgrid}.
\ifCLASSOPTIONcaptionsoff
  \newpage
\fi

\end{document}